\documentclass[a4,11pt]{article}
\usepackage[dvipdfmx]{graphicx}
\usepackage{lscape}
\usepackage{enumerate}
\usepackage{latexsym}
\usepackage{amsthm}
\usepackage{subcaption}  

\usepackage{color}
\usepackage{ulem}

\usepackage{framed}
\usepackage{color}

\newtheorem{theorem}{Theorem}

\newtheorem{lemma}{Lemma}

\usepackage{amsmath,amssymb}

\DeclareMathOperator{\argmin}{argmin}

\DeclareMathOperator{\sgn}{sgn}
\def\vec#1{\mbox{\boldmath $#1$}}

\begin{document}

\title{Structural grouping of extreme value models via graph fused lasso}

\author{
  {\sc Takuma Yoshida}$^{1,3}$,
  {\sc Koki Momoki}$^{2}$,
  {\sc Shuichi Kawano}$^{4}$\\
 {\it $^{1,2}$ Graduate School of Science and Engineering, Kagoshima University}\\
  {\it $^{3}$ Data Science and AI Innovation Research Promotion Center, Shiga University}\\
  {\it $^{4}$Faculty of Mathematics, Kyushu University}
}

\date{\empty}
\maketitle

\begin{abstract}
The generalized Pareto distribution (GPD) is a fundamental model for analyzing the tail behavior of a distribution.
In particular, the shape parameter of the GPD characterizes the extremal properties of the distribution.
As described in this paper, we propose a method for grouping shape parameters in the GPD for clustered data via graph fused lasso.
The proposed method simultaneously estimates the model parameters and identifies which clusters can be grouped together.
We establish the asymptotic theory of the proposed estimator and demonstrate that its variance is lower than that of the cluster-wise estimator.
This variance reduction not only enhances estimation stability but also provides a principled basis for identifying homogeneity and heterogeneity among clusters in terms of their tail behavior. 
We assess the performance of the proposed estimator through Monte Carlo simulations.
As an illustrative example, our method is applied to rainfall data from 996 clustered sites across Japan.
\end{abstract}


\section{INTRODUCTION}
In many fields of risk assessment, interest lies in predicting the tail behavior of a distribution rather than its mean or central characteristics.
However, tail prediction is inherently difficult and unstable because of the inherent data sparsity in extreme regions.
To construct efficient estimators and prediction models for tail behavior, we primarily examine the tail of the distribution rather than modeling the entire distribution.
Extreme value theory (EVT) provides a theoretical foundation and statistical tools for modeling tail behavior.
As a commonly used EVT model, the generalized Pareto distribution (GPD) is fitted to data exceeding a specified threshold and serves as a model for extreme values.
The GPD is characterized by two parameters: a shape parameter and a scale parameter.
The shape parameter, also known as the extreme value index, determines the heaviness of the tail and the domain of attraction, whereas the scale parameter describes the magnitude of threshold exceedance (de Haan and Ferreira 2006).
Estimation of shape and scale parameters is typically conducted using the maximum likelihood method (e.g., Coles 2001).
The fundamental theories of the maximum likelihood estimator were presented by Smith (1987) and Drees et al. (2004).

A typical example of a statistical problem using extreme value analysis is risk modeling of climate data (e.g., Bousquet and Bernardara 2021). 
In general, climate data are available as clustered data because these are observed daily at multiple observation sites. 
Then, the interesting purpose is the prediction of marginal distribution of each cluster and dependence between clusters. 
As described herein, we specifically examine the estimation of the marginal distribution of each cluster. 
Then, the cluster-wise estimation of the parameters included in GPD is a classical approach. 
However, for example, in climate data, one might consider that nearby clusters (sites) tend to have similar features of the data. 
This tendency motivates us to explore the possibility of constructing more accurate estimators of the parameters of GPD by incorporating mutual information between clusters. 
Actually, Hosking and Wallis (1997) provided regional frequency analysis, which are pooling data observed at several clusters where the statistical behavior of the data is assumed to be similar. 
Casson and Coles (1999) also studied the pooling method of information of clusters in spatial extreme value models. 
Einmahl et al. (2020) developed the novel extreme value theory for pooling clusters (and time dependence). 
However, with those approaches above, it is difficult to choose the clusters to be grouping when the number of sites is very large. 
Dupuis et al. (2023) and Rohrbeck and Tawn (2023) have also proposed methods for clustering in the context of extreme value analysis.
Dupuis et al. (2023) specifically examined block maxima using the generalized extreme value distribution and proposed a heuristic grouping method based on marginal similarity, selecting the best model by BIC. 
Their method might present scalability issues when the number of clusters is large because of the exponential growth in candidate models.
Rohrbeck and Tawn (2023) considered clustering sites with extreme value inference.
Their Bayesian approach jointly estimates marginal effect, dependence structure, and cluster assignments, incorporating similarity in shape and scale parameters as well as extreme dependence.
Details of other grouping methods in extreme value modeling are presented by Rohrbeck and Tawn (2023). 
Consequently, extreme value modeling with cluster grouping constitutes a useful framework for structurally characterizing tail behavior across clusters.

As described herein, we provide the novel method for simultaneous construction of accurate estimators of the parameters of the GPD in each cluster and grouping shape parameters in some clusters. 
Then, we aim to group the shape parameter between clusters because it affects the tail behavior and domain of attraction directly.
However, grouping of scale parameters should be treated carefully because the scale parameter behavior is closely related to both the threshold and the shape parameter (de Haan and Ferreira 2006, theorem 1.2.5). 
Consequently, when the scale parameter is grouped, we should consider a grouping threshold in addition to the shape parameter (Section 3.2 herein).
However, it is difficult because the determination of a threshold is beyond the framework of estimation of parameters in GPD. 
Therefore, we will try to group only shape parameters; not scale parameters. 
From this, the scale parameter represents the specific information of extreme value model for each cluster.
Our proposed method is conducted using a penalized maximum likelihood method. 
The likelihood part consists of the likelihood based on GPD for each cluster. 
In the part of the penalty, we use the fused lasso penalty for the shape parameters across clusters (Tibshirani et al. 2005). 
Then, if some of the shape parameters can be regarded as equal, then the difference between them is reduced to zero.
In this situation, we use the fused lasso on the graph of some connections between each pair of clusters. 
Statistical methods using the idea of graph fused lasso have been developed by She et al. (2010), Ke et al. (2015), Wang et al. (2016), Tang and Song (2016), Wang et al. (2018), and Hui et al. (2024). 
However, a study examining a grouping method of the extreme value feature using graph fused lasso has not been reported to date. 
By the sparsity of the fused lasso penalty, the shape parameters in multiple clusters are equal when these shape parameters are inherently similar. 
If the shape parameters in multiple clusters are detected to be equal, then such clusters are said to have common extreme value distribution. 
One important issue for applying the graph fused lasso is determination of which clusters are linked.
If the number of clusters is $J>2$, then the maximum number of possible connections (i.e., edges on the graph) between clusters is ${}_{J}\mathrm{C}_2=J(J-1)/2$.
For example, if $J=1000$, then we have $J(J-1)/2\approx 500,000$, which represents too many connections to optimize the estimation.  
She (2010) pointed out that too many penalties engender unstable estimation or infeasible computation in graph fused lasso problem. 
Therefore, for application, it is necessary to reduce the number of edges of the graph appropriately.
For that reason, in our approach, we use a user-specified graph of the connections between clusters by prior analysis or information of data.
For example, for the climate data or spatial data, we can consider the connection based on the distance between the locations of clusters.
For time series observed data, some dependence measure is also useful to decide the graph. 

After constructing the graph, we estimate the parameters using a penalized maximum likelihood approach with a graph fused lasso penalty.
To calibrate the penalty strength for each pair of clusters, we incorporate the adaptive lasso framework proposed by Zou (2006).
Specifically, cluster-wise estimators are used as initial values for the shape parameters.
Subsequently, we employ the first derivative of a non-concave penalty, such as SCAD (Fan and Li 2001) or MCP (Zhang 2010), to construct adaptive weights.
This weighting mechanism effectively removes the penalty for pairs of clusters for which initial estimates differ considerably, even if they are connected in the specified graph.
Consequently, by adaptively refining the graph structure through these weights, the proposed joint estimation and clustering procedure improves computational efficiency and decreases estimation bias arising from inappropriate penalization between heterogeneous clusters.
For this study, we establish the asymptotic theory of the proposed adaptive graph fused lasso estimator and demonstrate that its variance is low compared to that of the cluster-wise estimator.
This variance reduction not only enhances estimation stability but also facilitates a principled identification of homogeneity and heterogeneity among clusters in terms of their tail behavior.

This paper is organized as explained below. 
Section 2 presents reviews of the EVT and GPD. 
The proposed grouping method of clusters is described in Section 3. 
Theoretical properties of the proposed method are revealed in Section 4.
A simulation study is conducted in Section 5, whereas a real data example is demonstrated in Section 6. 
Section 7 presents the salient conclusions. 
In the Appendix, the technical lemmas and the proof of theorems in Section 4 are presented. 
Additional studies of simulation and real data example are presented as supplementary materials.

\section{GENERALIZED PARETO DISTRIBUTION}

Let $X_1,\ldots,X_n\in \mathbb{R}$ be $i.i.d.$ random variables with a distribution function $F(x)=P(X_i<x)$. 
First we define the domain of attraction of $F$. 
If there exists a constant $\gamma\in\mathbb{R}$ and sequences $a_n\in\mathbb{R}$ and $b_n>0, n\in\mathbb{N}$ such that 
\[
\lim_{n\rightarrow\infty}P\left(\frac{\max_i X_i-a_n}{b_n} <z\right)=G_\gamma(z)=\exp[-(1+\gamma z)_+^{-1/\gamma}],\ \ z>0,
\]
then the distribution function $F$ is said to belong to the maximum domain of attraction of $G_\gamma$, as denoted by $F\in{\cal D}(G_\gamma)$, where $(x)_+=\max\{x,0\}$.
The distribution function $G_\gamma$ is the so-called generalized extreme value distribution. 
We next introduce the GPD. 
Let $H(x\mid \gamma)= (1+\gamma x)_+^{-1/\gamma}$. 
If $\gamma=0$, then we define $H(x\mid 0)=\lim_{\gamma\rightarrow 0}H(x\mid\gamma)=1-e^{-y}$. 
The function $H(x\mid \gamma)$ is known as the survival function of the GPD.
According to Chapter 4 of Coles (2001), $F\in{\cal D}(G_\gamma)$ if and only if 
\begin{eqnarray}
\lim_{w\rightarrow x^*}P(X_i > w+ \xi_w x\mid X_i>w)
=
\lim_{w\rightarrow x^*}\frac{1-F(w+\xi_w x)}{1-F(w)}= H(x\mid \gamma) \label{GPD}
\end{eqnarray}
with some sequence $\xi_w>0$ of $w\in\mathbb{R}$ and $x^*=\sup\{x : F(x)<1\}$. 
In fact, $\xi_w$ is not unique because its behavior is determined only by the limiting property (de Haan and Ferreira 2006, theorem 1.2.5). 
From (\ref{GPD}), it is apparent that for the random variable $Y_i=X_i-w$ given some large $w\in\mathbb{R}$,
\[
P(Y_i> y\mid Y_i>0) \approx H\left(\frac{y}{\xi_w} \middle| \gamma\right).
\]
From the above, $\gamma$ and $\xi_w$ respectively represent the so-called shape parameter and scale parameter.
Consequently, the tail behavior of the distribution function $F$ can be dominated by $H(y/\xi_w\mid\gamma)$. 
For our analyses, we let
\[
h^o(y\mid\gamma, \xi_w)=\frac{\partial 1- H(y/\xi_w\mid\gamma)}{\partial y}=\frac{1}{\xi_w}\left(1+\frac{\gamma}{\xi_w} y\right)_+^{-1/\gamma-1}
\]
be the density function associated with $H(y/\xi_w\mid \gamma)$. 
Therefore, the log-likelihood function of $(\gamma,\xi_w)$ is defined as 
\[
\sum_{i=1}^n I(Y_i>0)\log h^o(Y_i\mid \gamma,\xi_w).
\]
In fact, the samples with $Y_i<0$ are removed from the log-likelihood function via the indicator function $I(\cdot)$ as non-extreme value data. 
By maximizing the log-likelihood, we obtain the maximum likelihood estimator of $(\gamma,\xi_w)$. 

According to Drees et al. (2004), in the method presented above, the estimators of $\gamma$ and $\xi_w$ are dependent. 
Chavez-Demoulin and Davison (2005) examined reparameterization of the model to orthogonalize the maximum likelihood estimators of shape and scale parameters.
Defining $\sigma_w=\xi_w(\gamma+1)$, we consider the pair $(\gamma,\sigma_w)$. 
Then, the survival function of the GPD is 
\[
H\left(\frac{y}{\xi_w}\mid \gamma\right)=H\left(\frac{(\gamma+1)y}{\sigma_w}\middle|\gamma\right).
\]
Therefore, the density function of $1-H$ with parameters $(\gamma,\sigma_w)$ becomes
\begin{eqnarray}
h(y\mid \gamma,\sigma_w) = \frac{\gamma+1}{\sigma_w}\left(1+\frac{\gamma(\gamma+1)}{\sigma_w} y\right)_+^{-1/\gamma-1}. \label{rescaleGP}
\end{eqnarray}
The log-likelihood function of $(\gamma,\sigma_w)$ is 
\[
\ell(\gamma,\sigma_w)=\sum_{i=1}^n I(Y_i>0)\log h(Y_i\mid \gamma,\sigma_w).
\]
The maximum likelihood estimator of $(\gamma,\sigma_w)$ is denoted by $(\tilde{\gamma},\tilde{\sigma}_w)$. 
By similar arguments to those presented by Smith (1987) and Drees et al. (2004), we can show that $(\tilde{\gamma}, \tilde{\sigma}_w)$ is distributed asymptotically as a normal distribution. 
Then, by reparameterization of the scale parameter, one can show that
\[
E\left[\frac{\partial^2 \log h(Y_i\mid \gamma,\sigma_w)}{\partial \gamma \partial \sigma_w} \middle| Y_i>0\right]=0.
\]
Showing this implies that the asymptotic covariance matrix of  $(\tilde{\gamma}, \tilde{\sigma}_w)$, which is obtained as the inverse of  Fisher-information matrix, is diagonal.
Consequently, $\tilde{\gamma}$ and $\tilde{\sigma}_w$ are asymptotically independent.
Throughout this paper, GPD is parameterized using $(\gamma, \sigma_w)$ instead of $(\gamma, \xi_w)$.

Finally, we review the limiting behavior of $\sigma_w$. 
When $\gamma<0$, $x^*$ is finite. 
The $x^*$ can be finite or infinite for $\gamma=0$. 
However, for the analyses presented in this paper, we assume that $H(x)=1-e^{-x/\sigma}$ with $\sigma>0$ if $\gamma=0$. 
Then, $x^*=\infty$ for also $\gamma=0$.
Under such a condition, according to Theorem 1.2.5 of de Haan and Ferreira (2006), we have 
\begin{eqnarray}
\left\{
\begin{array}{ll}
\displaystyle\lim_{w\rightarrow \infty} \frac{\sigma_w}{w} =\gamma(\gamma+1), & \gamma>0,\\
\displaystyle\lim_{w\rightarrow x^*} \frac{\sigma_w}{x^*-w} =-\gamma(\gamma+1), & \gamma<0,\\
\displaystyle\lim_{w\rightarrow \infty}  \sigma_w =\sigma, &\gamma=0.
\end{array}
\right.
\label{LimitSigma}
\end{eqnarray}
The formula (\ref{LimitSigma}) shows that $\sigma_w$ is not unique, but the behavior depends on the $\gamma$ and threshold $w$. 

\section{PROPOSED METHOD}

\subsection{Clustered data}

We let $\{(X_{i1},\ldots,X_{iJ}): i=1,\ldots,n\}$  be $J$-variate random variables from the joint distribution $F(x_1,\ldots,x_J) = P(X_{i1}<x_1,\ldots, X_{iJ}<x_J)$. 
Here, $X_{ij}, i=1,\ldots,n, j=1,\ldots,J$ is the $i$-th observation of $j$-th cluster.
That is, the $i$-th events for all clusters occur simultaneously. 
For example, $X_{ij}$ is the observation at $i$-th day of $j$-th location or site. 
Next, we assume that $X_{1j},\ldots,X_{n,j}$ are marginally distributed independently and identically as $F_j(x_j)=P(X_{ij}<x_j)$ for $j=1,\ldots,J$. 
In addition, we presume that $F_j\in{\cal D}(G_{\gamma_j}), j=1,\ldots,J$.
Similarly to Section 2, we can establish the tail model of $X_{ij}$ using the GPD for each cluster $j$. 
That is, for each $j$, given threshold $w_j\in\mathbb{R}$, $Y_{ij}=X_{ij}-w_j$ is assumed to be distributed as the GPD $P(Y_{ij}>y\mid Y_{ij}>0)\approx H((\gamma_j+1)y/\sigma_{j, w_j}\mid\gamma_j)$ with auxiliary sequence $\sigma_{j, w_j}>0$.

As described in this paper, we are interested in the extreme value inference of the marginal distributions for all clusters.  
Of course, it might be sufficient to fit the GPD separately for each cluster.  
However, for clustered data, constructing a model that incorporates or shares information across clusters is an important topic (Hosking and Wallis 1997; Rohrbeck and Tawn 2023).
In EVT, if the shape parameters are equal across clusters, i.e., $\gamma_j = \gamma_k = \gamma$ for $j \ne k$, then both $F_j$ and $F_k$ belong to the same domain of attraction, ${\cal D}(G_\gamma)$.  
This sameness of domain implies pooling information across clusters, which can lead to more stable estimates of the shape parameter and improve prediction of extreme value model. 
Therefore, we intend to group the clusters based on the similarity of their shape parameters $\gamma_1, \ldots, \gamma_J$, which aims to improve the marginal extreme value model for each cluster by incorporating information from other clusters.

In extreme value analysis for clustered data, the extreme dependence between clusters is also interesting topic. 
However, from copula theory, estimation of the dependence structure and the marginal effects are separable (Nelsen 2006; Genest and Segers 2009). 
Under this separability, the present study focuses on marginal inference for each cluster independently, without explicitly modeling the extreme dependence between clusters.

\subsection{Graph fused lasso penalization}

To group the shape parameters having common information, we use the graph fused lasso-type penalized (negative) log-likelihood method. 
Define $\vec{\gamma}=(\gamma_1,\ldots,\gamma_J)^\top$, $\vec{\sigma}=(\sigma_{1},\ldots,\sigma_{J})^\top$ and the negative log-likelihood function as
\[
\ell(\vec{\gamma},\vec{\sigma})=-\sum_{j=1}^J\sum_{i=1}^{n} I(Y_{ij}>0)\log h(Y_{ij}\mid\gamma_j,\sigma_{j}).
\] 
We next define the graph fused lasso penalty. 
We let $G=(V, E)$ be a graph with the vertices $V=\{1,\ldots,J\}$ and unordered edges $E=\{e_1,\ldots,e_M\}$ for $M>0$. 
The vertices correspond to the index of clusters. 
Edges have the role of connection by which clusters are linked.  
For any $1\leq m\leq M$, there exists $k, j\in V$ such that $e_m=(j,k)$. 
The loss function of graph fused lasso (e.g., Wang et al. 2016) is defined as
\begin{eqnarray}
\ell_F(\vec{\gamma}, \vec{\sigma})
&=&\ell(\vec{\gamma},\vec{\sigma})+\lambda\sum_{(j,k)\in E} v_{j,k}\lvert\gamma_j-\gamma_k\rvert, \label{GFLasso}
\end{eqnarray}
where $\lambda>0$ is the tuning parameter and $v_{j,k}$s are weights.
As described in this paper, the graph $G$ is user-specified. 
As the weight, the typical choice is $v_{j,k}=1$, which indicates the original fused lasso penalty (Tibshirani et al. 2005). 
The weight type of adaptive lasso (Zou 2006) uses $v_{j,k}=1/\lvert\tilde{\gamma}_j-\tilde{\gamma}_k\vert$, where $\tilde{\gamma}_j$ and $\tilde{\gamma}_k$ represent the cluster-wise estimates.
As described herein, we mainly use the derivative of folded-concave type penalty.
The weight is defined as $v_{j,k}=\rho_{\lambda,a}(\lvert\tilde{\gamma}_j-\tilde{\gamma}_k\rvert)$ where $\rho_{\lambda,a}(t), a>0$ is continuous except for the finite number of $t$ and satisfies $\rho_{\lambda,a}(t)=0$ for $\lvert t\vert>a\lambda$ and $\lim_{t\rightarrow 0+}\rho_{\lambda,a}(t)=1$.
As an example, the derivatives of SCAD (Fan and Li 2001) and MCP (Zhang 2010) are useful. 
The SCAD-type $\rho_\lambda$, which is the derivative of the original definition of the SCAD penalty, is defined as 
\[
\rho_{\lambda,a}(t)=\sgn(t) I(\lvert t\rvert<\lambda)+ \frac{(a\lambda-\lvert t\rvert)}{a-1}\sgn(t) I(\lambda<\lvert t\rvert<a\lambda)
\]
for the constant $a>0$. 
Here, $\sgn(x)$ is the signature of $x$. 
When $\lvert t\rvert>a\lambda$, we obtain {$\rho_{\lambda, a}(t)=0$, which implies that if $\lvert\tilde{\gamma}_j-\tilde{\gamma}_k\rvert>a\lambda$, then $v_{j,k}=0$. 
Then, the penalty of $j$-th and $k$-th cluster is removed. 
Therefore, these cannot be grouped as same value. 
Consequently, by our method, the graph $G$ is first set by users from pilot study or prior information of data.
Secondly, the graph validity is checked by the weight $\rho_{\lambda,a}(\cdot)$. 

The proposed method is fundamentally similar to the one-step local linear approximation method (Zou and Li 2008) using cluster-wise estimator as the initial estimator. 
Similarly to the original method of SCAD and that presented by Ke et al.\ (2015), we can consider the SCAD penalty and can update the penalty along with the estimator. 
However, for our experiments, the results obtained using iterative method were quite similar to those obtained using one-step optimization. 

Sharing information across clusters may stablize the extreme value model. 
However, in the proposed method, we do not consider grouping the scale parameters $(\sigma_1,\ldots,\sigma_J)$.  
From this, the difference between clusters can also be incorporated. 
In other words, if scale parameter also grouped, it loses the specific information of extreme value model for each cluster. 
Moreover, introducing an additional penalty for scale parameters complicates optimization of the objective function.  
For these reasons, we do not perform grouping of the scale parameters when using the proposed method.
In that sense, our method is related to the model worked by Einmahl et al. (2020).

\subsection{Tuning parameter selection}

We denote $\{(\hat{\gamma}_j, \hat{\sigma}_j) : j=1,\ldots, J\}$ as the minimizer of (\ref{GFLasso}). 
Using the proposed method (\ref{GFLasso}), the selection of the tuning parameter $\lambda$ is crucially important. 
For this study, we use BIC-type criteria as 
\begin{eqnarray*}
{\rm BIC}(\lambda) = -2 \sum_{j=1}^J\sum_{i=1}^{n} I(Y_{ij}>0)\log h(Y_{ij}\mid\hat{\gamma}_j, \hat{\sigma}_{j}) + (J+K(\lambda))\log\left(\sum_{j=1}^J n_j\right),
\end{eqnarray*}
where for each $j=1,\ldots,J$, $n_j =\sum_{i=1}^n I(Y_{ij}>0)$ denotes the effective sample size for the $j$-th cluster, and $K(\lambda)$ represents the number of groups conducted by the shape parameters. 
When $\lambda=0$, the grouping does not proceeded. Therefore, $K(0)=J$. 
If $\lambda\rightarrow\infty$, then all $\gamma_j$'s are grouped as one parameter, which indicates $K(\infty)=\kappa(G)$, where $\kappa(G)$ is the number of connected components from initial graph $G$. 
However, when $\lambda$ is very large, some clusters will be forced to be grouped irrespective of the goodness of fit to the likelihood.
In that sense, the BIC above balances the goodness of fit and the model complexity.
In fact, in the BIC above, the term $J\log\left(\sum_{j=1}^J n_j\right)$ is related to the number of scale parameters. It is meaningless to select $\lambda$. 
A similar type of BIC was used also by Dupius et al.\ (2023).

\subsection{Implementation}

To solve (\ref{GFLasso}), we use the alternating direction method of multipliers (ADMM). 
The general properties of ADMM are described by Boyd et al. (2011).
Letting $D$ be $M\times J$ oriented incidence matrix of the graph $G$ with elements $\{-1, 0, 1\}$, where $M$ is the number of edge $E$, then for $m=1,\ldots,M$, if $e_m=(j,k)$, the $m$-th row vector of $D$ consists $1$ for $j$-th element, - for the $k$-th element, and 0 otherwise, where the orientations of signs are arbitrary. 
We next define the $M$-diagonal matrix $V$ with $(m,m)$-elements $\tilde{v}_m = v_{j,k}$ for $e_m=(j,k)$. 
That is, we obtain
\[
V D\vec{\gamma} = (v_{j,k} (\gamma_j-\gamma_k), (j,k)\in E).
\]
Let $\|\cdot\|_1$ be $L_1$-norm. 
Then, the loss function of (\ref{GFLasso}) can be written as
\[
\ell_F(\vec{\gamma},\vec{\sigma}) = \ell(\vec{\gamma},\vec{\sigma})+ \lambda \|V D\vec{\gamma}\|_1.
\]
The equality-constrained optimization problem is 
\begin{eqnarray*}
&&{\rm minimize}\ \ \ell(\vec{\gamma},\vec{\sigma})+ \lambda \|V \vec{u}\|_1\\
&&{\rm subject\ to}\ \ D\vec{\gamma} = \vec{u}. 
\end{eqnarray*}
Its augmented Lagrangian function is obtainable as
\begin{eqnarray*}
\ell_F(\vec{\gamma},\vec{\sigma},\vec{u},\vec{s})
=
\ell(\vec{\gamma},\vec{\sigma}) + \lambda \|V \vec{u}\|_1 +\vec{s}^\top(D\vec{\gamma}-\vec{u})+\frac{\rho}{2}\|D\vec{\gamma}-\vec{u}\|^2
\end{eqnarray*}
with the unknown vector $\vec{u}, \vec{s}\in\mathbb{R}^{M}$ and the additional tuning parameter (step size) $\rho>0$. 
The optimization problem of ADMM can be established as shown below. 
Letting $(\vec{\gamma}^{(0)}, \vec{\sigma}^{(0)})$ be the initial estimator of $(\vec{\gamma}, \vec{\sigma})$, then the estimators $(\vec{\gamma}^{(t)}, \vec{\sigma}^{(t)}, \vec{u}^{(t)}, \vec{s}^{(t)})$ of parameters for $t$-step iteration are
\begin{eqnarray}
 \vec{u}^{(t)}
 &=& \underset{\vec{u}}{\argmin}\ \     \{\vec{s}^{(t-1)}\}^\top(D\vec{\gamma}^{(t-1)}-\vec{u})+\frac{\rho}{2}\|D\vec{\gamma}^{(t-1)}-\vec{u}\|^2 +\lambda \|V \vec{u}\|_1,\nonumber\\
 \vec{s}^{(t)}
 &=&
\vec{s}^{(t-1)} +\rho(D\vec{\gamma}^{(t-1)}-\vec{u}^{(t)}),\nonumber\\
(\vec{\gamma}^{(t)}, \vec{\sigma}^{(t)})
&=& \underset{(\vec{\gamma},\vec{\sigma})}{\argmin}\ \   \ell(\vec{\gamma},\vec{\sigma})   +\{\vec{s}^{(t)}\}^\top(D\vec{\gamma}-\vec{u}^{(t)})+\frac{\rho}{2}\|D\vec{\gamma}-\vec{u}^{(t)}\|^2. \label{lossPara}
\end{eqnarray}
Each element of $\vec{u}^{(t)}$: $u_m^{(t)}$ is obtainable by closed form as 
\[
u_m^{(t)} = \sgn\left((D\vec{\gamma}^{(t-1)})_m +\frac{s_m^{(t-1)}}{\rho}\right)\left(\left|(D\vec{\gamma}^{(t-1)})_m +\frac{s_m^{(t-1)}}{\rho}\right|-\frac{\lambda w_{m}}{\rho}\right)_+,\ \ m=1,\ldots, M,
\]
where $(D\vec{\gamma}^{(t-1)})_m$ and $s_m^{(t-1)}$ respectively represent the $m$-th element of $(D\vec{\gamma}^{(t-1)})$ and $\vec{s}^{(t-1)}$.

In our method, step size $\rho$ is varied in each iteration as $\rho=\rho^{(t)}$ because the estimator is heavily dependent on the choice of $\rho$. 
The step size $\rho^{(t)}$ is defined by the method of He et al. (2000) and Wang and Liao (2001), defined as 
\[
\rho^{(t)}
=
\left\{
\begin{array}{lc}
2 \rho^{(t-1)}, & \|\vec{\eta}^{(t-1)}\|>\mu \| \vec{\xi}^{(t-1)}\|,\\
2^{-1}\rho^{(t-1)}, & \|\vec{\xi}^{(t-1)}\| > \mu \|\vec{\eta}^{(t-1)}\|,\\
\rho^{(t-1)}, & {\rm otherwise},
\end{array}
\right.
\]
where $\vec{\eta}^{(t)} = D\vec{\gamma}^{(t)}-\vec{u}^{(t)}$, $\vec{\xi}^{(t)}= \rho^{(t-1)} D(\vec{\gamma}^{(t)}-\vec{\gamma}^{(t-1)})$ and $\mu>0$. 
For Sections 5 and 6, we used $\mu=5$.  
Roughly speaking, if the norm $\|\vec{\eta}^{(t-1)} \|$ is large, then the constraint $D\vec{\gamma} = \vec{u}$ is not well satisfied, implying insufficient sparsity in $D\vec{\gamma}$. Therefore, the last term of (\ref{lossPara}) is emphasized by increasing $\rho^{(t)}$.
However, a large norm $\|\vec{\xi}^{(t-1)}\|$ implies poor goodness-of-fit to the model. In this case, decreasing $\rho^{(t)}$ reduces the influence of the last term of (\ref{lossPara}).
Consequently, dynamically updating the step size $\rho^{(t)}$ balances the sparsity of $D\vec{\gamma}$ and the goodness-of-fit of the model.

If the algorithm is stopped for $t$-th iteration, then we define the estimator $(\hat{\vec{\gamma}},\hat{\vec{\sigma}})=(\vec{\gamma}^{(t)}, \vec{\sigma}^{(t)})$. 
The grouped clusters are identifiable by the zero elements of $\hat{\vec{u}}=D\hat{\vec{\gamma}}$.

\section{ASYMPTOTIC THEORY}

\subsection{General conditions}

To investigate the theoretical properties of the proposed estimator, some technical conditions must be established. 
The general conditions are listed as follows to establish the asymptotic theory of the estimator. 
\begin{enumerate}
\item[(C1)] There exist constants $\gamma_*,\gamma^*\in(-1/2,\infty)$ such that $\gamma_*< \min_j \gamma_j < \max_j \gamma_j <\gamma^*$. 
\item[(C2)] For $j=1,\ldots,J$, $w_j\rightarrow x_j^*$, $n_j\rightarrow \infty$ and $n_j/n\stackrel{P}{\rightarrow} 0$ as $n\rightarrow\infty$. 
\end{enumerate}

(C1) is the standard condition of the GPD. 
If the shape parameter is less than $-1/2$, then the Fisher information matrix based on GPD is well known not to be positive definite, which implies that the shape parameter estimator is unstable. 
Consequently, $\gamma_*>-1/2$ is important.  
(C2) is also standard in the asymptotic analysis in the EVT. 
Actually from (C2), the effective sample size $n_j$ can be regarded as the intermediate order sequence in the EVT (de Haan and Ferreira 2006, chap. 2). 

In addition, we state the convergence rate of the distribution function converges to GPD. 
for $F_j\in{\cal D}(G_{\gamma_j})$, we now assume that, for the sequence of threshold $w_{j}$, there exists an auxiliary function $\alpha_j(w_j)$ satisfying $\alpha_j(w_j)\rightarrow 0$ as $w_j\rightarrow x_j^*:=\sup\{x: F_j(x)<1\}$ such that for any $x>0$,
\begin{eqnarray}
\lim_{w_j\rightarrow x_j^*}\left|\frac{\frac{1-F_j(w_j+x)}{1-F_j(w_j)}- H((\gamma_j+1)x/\sigma_{j,w_j} \mid \gamma_j)}{\alpha_j(w_j)}\right| = O(1). \label{SecondOrderSimple}
\end{eqnarray}
The condition (\ref{SecondOrderSimple}) is detailed as (\ref{SecondOrder}) given in Appendix. 
This condition is closely related to the second-order condition in Extreme Value Theory (EVT) (e.g., de Haan and Ferreira, 2006, Chapter 2) and is essential for characterizing the asymptotic bias of estimators in extreme value models. 
However, since the primary focus of this study is not bias reduction, we employ the simplified notation above. 
Specifically, in Theorems \ref{Oracle} and \ref{MainTheorem}, we adopt a variance-dominated assumption, where the term $\alpha_j(w_j)$ is assumed to vanish at a sufficiently fast rate.

\subsection{Oracle estimator}

To clarify benefits of the proposed grouping shape parameter method, we first show the asymptotic result of the estimator of shape parameter under the condition that grouping is already done.
Let ${\cal A}$ be the index set of one of the true groups. 
For simplicity, we write $\gamma=\gamma_j,\ j\in{\cal A}$.
That is, for all $j\in {\cal A}$, $F_j\in{\cal D}(G_\gamma)$. 
Without loss of generality, we can assume that ${\cal A}=\{1,\ldots,J\}$.
We then construct the estimator $\{\hat{\gamma},\hat{\sigma}_j: j\in{\cal A}\}$ as the maximizer of
\begin{eqnarray*}
\ell_{{\cal A}}(\gamma,\sigma_1,\ldots,\sigma_J) = \sum_{j\in{\cal A}} \sum_{i=1}^nI(Y_{ij}>\omega_j)\log h(Y_{ij}\mid\gamma, \sigma_{j}).
\end{eqnarray*}

Letting $n_{{\cal A}} =\sum_{j\in{\cal A}} n_j$, we can then obtain the following theorem.

\begin{theorem}\label{Oracle}
Suppose that {\rm (C1)} and {\rm (C2)}, and $n_{\cal A}^{1/2}\max_{j\in{\cal A}}\alpha_j(w_j)\rightarrow 0$.
Then, as $n\rightarrow \infty$,
\[
\{n_{{\cal A}}\}^{1/2}(\hat{\gamma}-\gamma)\xrightarrow{D} N(0, (\gamma+1)^2)
\]
and 
\[
n_j^{1/2}\left(\frac{\hat{\sigma}_j}{\sigma_{j,w_j}}-1\right)\xrightarrow{D} N(0, 2\gamma+1),\ \ j\in{\cal A}.
\]
Furthermore, for any $j\in{\cal A}$, $\{n_{{\cal A}}\}^{1/2} Cov(\hat{\gamma},\hat{\sigma}_j)\rightarrow 0$ as $n\rightarrow \infty$.
\end{theorem}

From Theorem \ref{Oracle}, is readily apparent that the rate of convergence of the estimator $\hat{\gamma}$ is higher than that of each $\tilde{\gamma}_j, j\in{\cal A}$ if ${\cal A}$ contains at least two indices of clusters. 
Consequently, by grouping shape parameters among clusters, the accuracy of the estimator of the shape parameter is increasing drastically. 
We also find that the grouping shape parameter does not affect to the asymptotic rate of the scale parameter.
However, to obtain the confidence interval of $\sigma_j$, it is necessary to estimate $\gamma$, which appear as the standard deviation of $\hat{\sigma}_j$. 
Therefore, the accurate estimation of $\gamma$ also has a good influence on prediction of the scale parameter.

We briefly describe the relation between the grouped estimator $\hat{\gamma}$ and the cluster-wise estimators $\tilde{\gamma}_j$, $j \in \mathcal{A}$.
The likelihood function of the GPD is well known to be asymmetric with respect to the shape parameter (and also with respect to the scale parameter).
Specifically, for each $j$, the likelihood decreases much more rapidly when the shape parameter moves toward smaller values than when it moves toward larger values around the MLE $\tilde{\gamma}_j$.

As a consequence, when shape parameters are grouped across clusters, the grouped estimator tends to be shifted toward larger values to avoid a great loss in likelihood.
In other words, the set $\{ j \in \mathcal{A} : \hat{\gamma} > \tilde{\gamma}_j \}$ tends to be large.
This mechanism leads to a more conservative model for many clusters.

In application, it is interesting to investigate the probability of a rare event occurring, $P(X_{ij}>x)$ for $x >w_j$, instead of the estimator of parameters itself. 
From GPD modeling, we have $P(X_{ij}>x)= P(X_{ij}>x\mid X_{ij}>w_j)P(X_{ij}>w_j)\approx H(x-w_j\mid \gamma_j,\sigma_j) P(X_{ij}>w_j)$. 
Therefore, we can evaluate $P(X_{ij}>x)$ by replacing $(\gamma,\sigma_j)$ with its estimator. 
We are also interested in evaluating the quantile of $X_{ij}$ at region of extreme value, which is called the return level. 
For any $\tau\in(0,1)$, the $\tau$-th return level of $X_{ij}$ is defined as $R(\tau)= \inf\{t : P(X_{ij}>t) < \tau\}$.
If $\tau < P(X_{ij}>w_j)$, then the continuity of $H$ yields that 
\[
R(\tau) = R(\tau\mid\gamma,\sigma_j) \approx \{t :  H(y\mid\gamma,\sigma_j)P(X_{ij}>w_j) = \tau\} =w_j + \frac{\sigma_j}{\gamma(\gamma+1)} \left\{ \left( \frac{\tau }{P(X_{ij}>w_j)}\right)^{-\gamma} -1\right\}.
\]
We write $\xi_j =P(X_{ij}>w_j)$. 
Then, we obtain $\hat{\xi}_j=n_j/n$ and $\sqrt{n}(\hat{\xi}_j-\xi_j)\stackrel{D}{\to} N(0,\xi(1-\xi))$. 
Because $R$ depends on $(\gamma,\sigma_j, \xi_j)$, the return level can be estimated by replacing unknown parameters with these estimators.
By the delta method, we obtain 
\begin{eqnarray*}
&&R(\tau\mid\hat{\gamma},\hat{\sigma}_j, \hat{\xi}_j)\\
&&\approx R(\tau\mid\gamma,\sigma_j, \xi_j) + R_\gamma(\tau\mid\gamma,\sigma_j, \xi_j) (\hat{\gamma}-\gamma) + R_\sigma(\tau\mid\gamma,\sigma_j, \xi_j) (\hat{\sigma}_j-\sigma_j)\\
&&\quad + R_{\xi}(\tau\mid\gamma, \sigma_j, \xi_j)(\hat{\xi}_j-\xi_j),
\end{eqnarray*}
where $R_\gamma(\tau\mid\gamma,\sigma_j, \xi_j) = \partial R(\tau\mid\gamma,\sigma_j,\xi)/\partial \gamma$, $R_\sigma(\tau\mid\gamma,\sigma_j,\xi_j)= \partial R(\tau\mid\gamma,\sigma_j,\xi_j)/\partial \sigma_j$, and $R_\xi(\tau\mid\gamma,\sigma_j, \xi_j) = \partial R(\tau\mid\gamma,\sigma_j,\xi)/\partial \xi_j$ (see Section 4.3.3 of Coles 2001). 
Letting $z_p$ be the $100(1-p)\%$ quantile of standard normal distribution such that $P(Z>z_p)=p$ for $Z\sim N(0,1)$, then from Theorem \ref{Oracle}, the $100 (1-p)\%$ confidence interval of $R(\tau\mid\hat{\gamma}, \hat{\sigma}_j)$ is obtained as
\begin{eqnarray}
&&R(\tau\mid\hat{\gamma}, \hat{\sigma}_j)\pm z_{p/2}\sqrt{\frac{\hat{R}_\gamma^2(\hat{\gamma}_j+1)^2}{n_{{\cal A}}}+\frac{\hat{R}_\sigma^2 \hat{\sigma}_j^2(2\hat{\gamma}+1)}{n_j} + \frac{\hat{R}_\xi^2\hat{\xi}_j(1-\hat{\xi}_j)}{n} }, \label{returnlevel}
\end{eqnarray}
where $\hat{R}_\gamma = R_\gamma(\tau\mid\hat{\gamma},\hat{\sigma}_j,\hat{\xi}_j)$, and definitions of $\hat{R}_\sigma, \hat{R}_\xi$ are similar. 
When the shape parameters are not grouped, the term $n_{{\cal A}}$ in (\ref{returnlevel}) is replaced by $n_j$.
Consequently, the confidence interval obtained from the grouped estimator becomes shorter than that from the cluster-wise estimator, indicating that the estimation error of the return level is reduced.
Actually, the return level $R(\tau\mid \gamma, \sigma)$ depends exponentially on the shape parameter $\gamma$, whereas it depends linearly on the scale parameter $\sigma$.  
As a result, the sensitivity of $R(\tau\mid \gamma, \sigma_j)$ to $\gamma$ is much greater than to $\sigma_j$. 
In (\ref{returnlevel}), the rate of $\hat{\gamma}$, $n_{\mathcal{A}}$, which is larger than $n_j$ (rate of cluster-wise estimator), helps stabilize the behavior of $R_\gamma$.

Theorem \ref{Oracle} is useful if grouping is correct. 
However, Theorem \ref{Oracle} is only a theoretical result because we are unable to detect ${\cal A}$ completely in practice. 
As explained in the following section, we establish the asymptotic results obtained using the proposed estimator that creates a new group of clusters by graph fused lasso.

\subsection{Graph fused lasso estimator}

Let ${\cal A}_1,\ldots, {\cal A}_K$ be the index sets of true groups of shape parameters such that for $k=1,\ldots,K$, if $\alpha,\beta\in {\cal A}_k$, we obtain $\gamma_\alpha=\gamma_\beta=:\gamma$ and $F_\alpha,F_\beta\in{\cal D}(G_\gamma)$. 
When ${\cal A}_j=\{j\}$ and $K=J$, the distributions for all clusters belong to a domain of attraction with different shape parameters. 
The cluster-wise estimator is obtained under this condition. 
As described in this paper, we generally assume that $K<J$.  
The penalty in (\ref{GFLasso}) can be expressed as
\[
\lambda\sum_{(j,k)\in E}v_{j,k}|\gamma_j-\gamma_k| = \frac{\lambda}{2}\sum_{j,k=1}^J c_{j,k}v_{j,k}|\gamma_j-\gamma_k|
\]
where $c_{j,k}=1$ if there exists $m$ such that $e_{m}=(j,k)$, and $c_{j,k}=0$ if there is no edge between indices $j$ and $k$.
Therefore, we obtain $c_{j,k}=c_{k,j}$ and $c_{j,j}=0$ for $j,k=1,\ldots,J$. 
The term 1/2 above is absorbed in the tuning parameter $\lambda$.

Let $G_k=\{{\cal A}_k, E_k\}$ be the subgraph induced by the set of nodes ${\cal A}_k$, where $E_k=\{(i, j)\in{\cal A}_k\times{\cal A}_k: c_{i, j}=1\}.$
Here, $G$ and $G_k$ are undirected graphs. 
Then the sets ${\cal B}_{k,\ell}$ for $\ell=1,\ldots,L_k$ are defined as presented below.
\begin{enumerate}
\item ${\cal B}_{k,1},\ldots,{\cal B}_{k,L_k}$ are the distinct connected components of the subgraph $G_k$.
\item ${\cal B}_{k,1},\ldots,{\cal B}_{k,L_k}$ satisfy
\[
\bigcup_{\ell=1}^{L_k} {\cal B}_{k,\ell} ={\cal A}_k\ \ {\rm and}\ \ {\cal B}_{k,s}\cap{\cal B}_{k,t}=\emptyset\ {\rm for}\ s\not=t.
\]
\end{enumerate}
For example, we assume that ${\cal A}_1=\{1,2,3,4,5,6\}$ and $c_{1,2}=c_{2,3}=1$, $c_{1,3}=0$, $c_{4,5}=1$, $c_{j,4}=c_{j,5}=0$ for $j=1,2,3$, and $c_{i,6}=0$ for $i=1,\ldots,5$. 
Then, the graph is $G_1=({\cal A}_1, E_1)$ with $E_1=\{(1,2), (2,3), (4,5)\}$. Subsequently, we can construct ${\cal B}_{1,1}=\{1,2,3\}$, ${\cal B}_{1,2}=\{4,5\}$ and ${\cal B}_{1,3}=\{6\}$. 
It is noteworthy in this case that $\gamma_1=\cdots=\gamma_6$ although ${\cal B}_{1,i}\cap {\cal B}_{1,j}=\emptyset$ for $i,j=1,2,3\ (i\not=j)$.

The objective loss function (\ref{GFLasso}) can be written as 
\begin{eqnarray}
\ell_F(\vec{\gamma}, \vec{\sigma})
&=&\ell(\vec{\gamma},\vec{\sigma})+\lambda \sum_{k=1}^K\sum_{\ell=1}^{L_k}\sum_{i,j\in{\cal B}_{k,\ell}}c_{i,j}v_{j,k}|\gamma_i-\gamma_j| +\lambda \sum_{(i,j)\in{\cal A}^*}c_{i,j}v_{j,k}|\gamma_i-\gamma_j|, \label{GFLassoTheory}
\end{eqnarray}
where ${\cal A}^*=\{(\alpha,\beta) | {}^\exists k,h\ (k\not=h)\ {\rm s.t.}\ \alpha\in{\cal A}_{k}\ {\rm and}\ \beta\in{\cal A}_h\}$. 
For $(i,j)\in{\cal A}^*$ if $c_{i,j}=1$, meaning that the link of $\gamma_i$ and $\gamma_j$ is mis-specified. 
Because the graph is user-specified, it does not guarantee that clusters sharing the same value of shape parameter are linked correctly.
In other words, it is difficult to produce a graph to connect all clusters in ${\cal A}_k$ for $k=1,\ldots,K$. 
Therefore, the set ${\cal A}^*$, which indicates the misconnection, is generally non-empty. 
If we want to make the graph so that $L_k=1$ and ${\cal B}_{k,1}={\cal A}_k$ for $k=1,\ldots,K$, then we have no choice but to consider a complete graph that includes all possible pairs of clusters. 
However, the complete graph engenders heavy tasks of computation (She 2010), rendering its use unrealistic in practice.

We further let $\vec{\gamma}_{{\rm oracle}}=(\gamma_{{\cal B}_{k,\ell}}, \ell=1,\ldots,L_k, k=1,\ldots,K)$.
Then, the loss function for oracle situation is
\begin{eqnarray}
\ell_{{\rm oracle}}(\vec{\gamma}_{{\rm oracle}},\vec{\sigma})
&=&-\sum_{k=1}^K \sum_{\ell=1}^{L_k}\sum_{j\in{\cal B}_{k,\ell}} \sum_{i=1}^{n_j} \log h(Y_{ij}\mid\gamma_{{\cal B}_{k,\ell}},\sigma_{j}).
\end{eqnarray}
Then, the estimator obtained by minimizing $\ell_{{\rm oracle}}$ is defined as 
\[
\hat{\vec{\gamma}}_{{\rm oracle}}=(\hat{\gamma}_{{\cal B}_{k,\ell}}, \ell=1,\ldots,L_k,k=1,\ldots,K)
\]
and 
\[
\hat{\vec{\sigma}}_{{\rm oracle}}=(\hat{\sigma}_{{\cal B}_{k,\ell}, j}, j\in {\cal B}_{k,\ell}, \ell=1,\ldots,L_k,k=1,\ldots,K)
\]
The asymptotic property of $(\hat{\gamma}_{{\cal B}_{k,\ell}}, \hat{\sigma}_{{\cal B}_{k,\ell}, j})$ is obtained using theorem \ref{Oracle} with the true set ${\cal B}_{k,\ell}$.

Whether the proposed estimator attains the oracle estimator or not is important. 
The following theorem shows such an efficient property of the estimator. 
Before stating the theorem, we consider the following conditions.

\begin{itemize}
\item[(C3)]  There exists $\delta>0$ such that $\min_{j\not= k}\min_{\alpha\in{\cal A}_j,\beta\in{\cal A}_k}|\gamma_\alpha-\gamma_\beta|>\delta$.
\item[(C4)] The weight function $p_{\lambda,a}(\cdot)$ is a monotonically decreasing function, $\lim_{t\rightarrow +0}p_{\lambda,a}(t)=1$ and $p_{\lambda,a}(t)=0$ for $t\geq a\lambda$ with $\lambda,a>0$. 
\end{itemize}

We define $n_{{\cal B}_{k,\ell}}=\sum_{j\in{\cal B}_{k,\ell}} n_j$ for $\ell=1,\ldots,L_k, k=1,\ldots,K$.

\begin{theorem}\label{MainTheorem}
Suppose that {\rm (C1)}--{\rm (C4)}, and for $\ell=1,\ldots,L_k, k=1,\ldots,K$, $n_{{\cal B}_{k,\ell}}^{1/2}\max_{j\in{\cal B}_{k,\ell}}
\alpha_j(w_j) \rightarrow 0$ as $n\rightarrow\infty$. 
Furthermore, suppose that $\lambda$ satisfies $n\lambda\rightarrow\infty$ and $\max_{k,\ell}|{\cal B}_{k,\ell}|/(n\lambda)\rightarrow 0$. 
Then, $\ell_F(\vec{\gamma},\vec{\sigma})$ has a strictly local minimizer $(\hat{\vec{\gamma}},\hat{\vec{\sigma}})$ such that 
\[
\hat{\vec{\gamma}}=\hat{\vec{\gamma}}_{{\rm oracle}}, \ \ \hat{\vec{\sigma}}=\hat{\vec{\sigma}}_{{\rm oracle}}
\]
with probability tending to one. 
\end{theorem}

From Theorems \ref{Oracle} and \ref{MainTheorem}, it is readily apparent that the asymptotic rate of the variance of the proposed estimator is faster than that of the cluster-wise estimator. 
This result is affected strongly by the condition that $n_{{\cal B}_{k,\ell}}^{1/2}\max_{j\in{\cal B}_{k,\ell}}
\alpha_j(w_j) \rightarrow 0$ in Theorem \ref{MainTheorem}. 
This assumption indicates that the asymptotic bias of the estimator is sufficiently small. 
One might think that this condition is strong, but this is justified in the context of our grouping method. 
In practice, if the large value of the threshold is used, then the bias becomes small. 
For the cluster-wise estimator, the very large value of threshold leads to a large variance because the effective sample size becomes small. 
By contrast, using our method, such a difficulty of the small effective sample size would be covered by grouping clusters.
For example, the confidence interval of the parameters in GPD and return level is constructed traditionally under the condition that the bias is assumed to be zero (Coles 2001, sect. 3).
The grouping of clusters in our method helps to construct the accurate confidence interval.

\section{SIMULATION}

We conducted the Monte Carlo simulation to investigate the finite sample performance of the proposed estimator. 
The advantage of the proposed grouping method is reducing the variance of the estimator compared with cluster-wise estimator.
To confirm this benefit, in this simulation, we generate data from each cluster using the GPD with some dependence between clusters. 
The data were generated according to the following procedure.
First, we independently obtained $Z_{1,1},\ldots,Z_{n,1}\sim N(0,1)$, which are used to create the data for the first cluster.
Next, the $j$-th random vector $(Z_{1, j},\ldots, Z_{n,j})$ was determined depending on the $(j-1)$-th random vector $(Z_{1, j},\ldots, Z_{n,j})$ as $Z_{i,j} = \rho Z_{i,j-1} + \sqrt{1-\rho^2} V_i$, where $\rho=0.999$ and $V_i\overset{i.i.d.}{\sim} N(0, 1)$.
Using $\{Z_{ij}, i=1, \ldots, n, j=1, \ldots, J\}$, finally, we made $U_{i,j}=\Phi(Z_{i,j})$ and $X_{i,j}= F_j^{-1}(U_{i,j})$, where $\Phi(\cdot)$ is the distribution function of a standard normal, and $F_j$ is defined as
\[
F_j(x)=\
\left\{
\begin{array}{cc}
1-\left(1 + \frac{\gamma_j(\gamma_j+1) x}{\sigma_j}\right)_+^{-1/\gamma_j}, & \gamma_j\not=0,\\
1-e^{-x/\sigma_j},& \gamma_j=0
\end{array}
\right.
\]
with
\[
\gamma_j= 0.3 - 0.05 \left(\left\lceil \frac{j}{100} \right\rceil -1\right)  j=1,\ldots,J,
\]
and 
\[
\sigma_j =\begin{cases}
40 - 5\left(\left\lfloor \frac{(j-1)\ {\rm mod}\ 100}{20}\right\rfloor \right), & \ j=1,\ldots,600,\\
40,& j=601,\ldots,700,\\
200+ 50\left(\left\lfloor \frac{(j-1)\ {\rm mod}\ 100}{20}\right\rfloor \right),& j=701,\ldots, J.
\end{cases}
\]
For this study, we set $n=120$ and $J=1100$ respectively as the sample size and the number of clusters. 
For $j=1,\ldots,J$, $X_{1,j},\ldots,X_{n,j}$ were independently distributed as GPD with parameter $\gamma_j$ and $\sigma_j$. 
The true structures of parameters are detailed below. 
Roughly speaking, in the model above, the true shape parameters were constructed by assigning blocks of 100 values, each decreasing by 0.05 starting from 0.3. 
For example, the first 100 values ($j=1,\ldots,100$) are 0.3. The next 100 ($j=101,\ldots,200$) are 0.25. Also, $j\in[ 501, 600]$ corresponds to $\gamma_j = 0$, $j\in[ 601, 700]$ corresponds to $\gamma_j = -0.05$, and the final 100 values $(j = 1001,\ldots,1100)$ correspond to $\gamma_j = -0.2$.
Consequently, 100 clusters share the same shape parameter. 
However, within each group of clusters sharing the same shape parameter, the scale parameters $\sigma_j$ varied every 20 clusters.
The correlation between $X_{i,j}$ and $X_{i,k}$ is approximately $\rho^{|j-k|}$. 
The pair of Gaussian random variables is well known to have no extremal dependence unless correlation $\rho$ is close to one (Coles 2001). 
Thus, the extreme dependence of $X_{i,j}$ and $X_{i,k}$ is not strong from a theoretical perspective. 
However, since $\rho=0.999$, large dependence can be found when $|j-k|$ is small as the finite sample performance.

Because the true model is GPD, the threshold selection is no longer needed. 
Therefore, the approximation bias, denoted as $\alpha_j(w_j)$ in Section 4.1, is zero. 
In other words, we set zero as the threshold for all clusters. 
Under this, we fitted the GPD to the clustered data. 
We denote $\tilde{\gamma}_j$ as the cluster-wise estimator of $\gamma_j$ for $j=1, \ldots, J$. 
The cluster-wise estimator is apparent as the baseline estimator. 
To apply our method, we must set the initial graph to add the fused lasso penalty. 
For this study, the graph edge is set as $(j,j+1), (j,j+2), (j,j+3)$, and $(j,j+4)$ for $j=1,\ldots, J-4$. 
The total number of initial edges was 4384. 
For example, for $1\leq j\leq 96$, the connection of the edge is correct because $\gamma_1=\cdots=\gamma_{100}$. 
However, for $98\leq j\leq 101$, there is the mis-specified edge because $\gamma_{100}\not=\gamma_{101}$. 
Consequently, there are 4284 correct edges and 100 incorrect edges in this graph.
Under this setting, we applied the proposed method to the data. 
The proposed estimator of the shape parameter is denoted as $\hat{\gamma}_j$. 
We evaluated the accuracy of the estimators for each cluster by the ratio of the mean squared error
\[
{\rm ratio\ of\ MSE} = \frac{E[(\hat{\gamma}_j-\gamma_j)^2]}{E[(\tilde{\gamma}_j-\gamma_j)^2]}
\]

Hereinafter, we designate GFL as the proposed method using graph fused lasso. 
The left panel of Figure \ref{FigSection5_1} depicts the behaviors of the cluster-wise estimator and the GFL for $j=1,\ldots, J$. 
The solid and dashed lines respectively represent median and lower/upper 5\% quantiles calculated based on 1000 Monte Carlo iterations. 
It is readily apparent that the performances of the median of two estimators are similar. 
We see that the bias remains for all clusters, but this is natural for not large sample size in extreme value analysis. 
However, the lower quantiles of the GFL are close to the true value compared with the cluster-wise estimator. 
This implies that the proposed method improves the stability of the estimator. 

The right panel of Figure \ref{FigSection5_1} presents the ratio of MSE for all clusters. 
One can say that the GFL performed better than the cluster-wise method if this ratio is less than 1.
For many clusters, the GFL improved the performance.  
While both estimators exhibit some bias appered in left panel, the results in the right panel demonstrate that the GFL effectively reduces variance compared to the cluster-wise estimator. 

\begin{figure}
\begin{center}
\includegraphics[width=150mm,height=80mm]{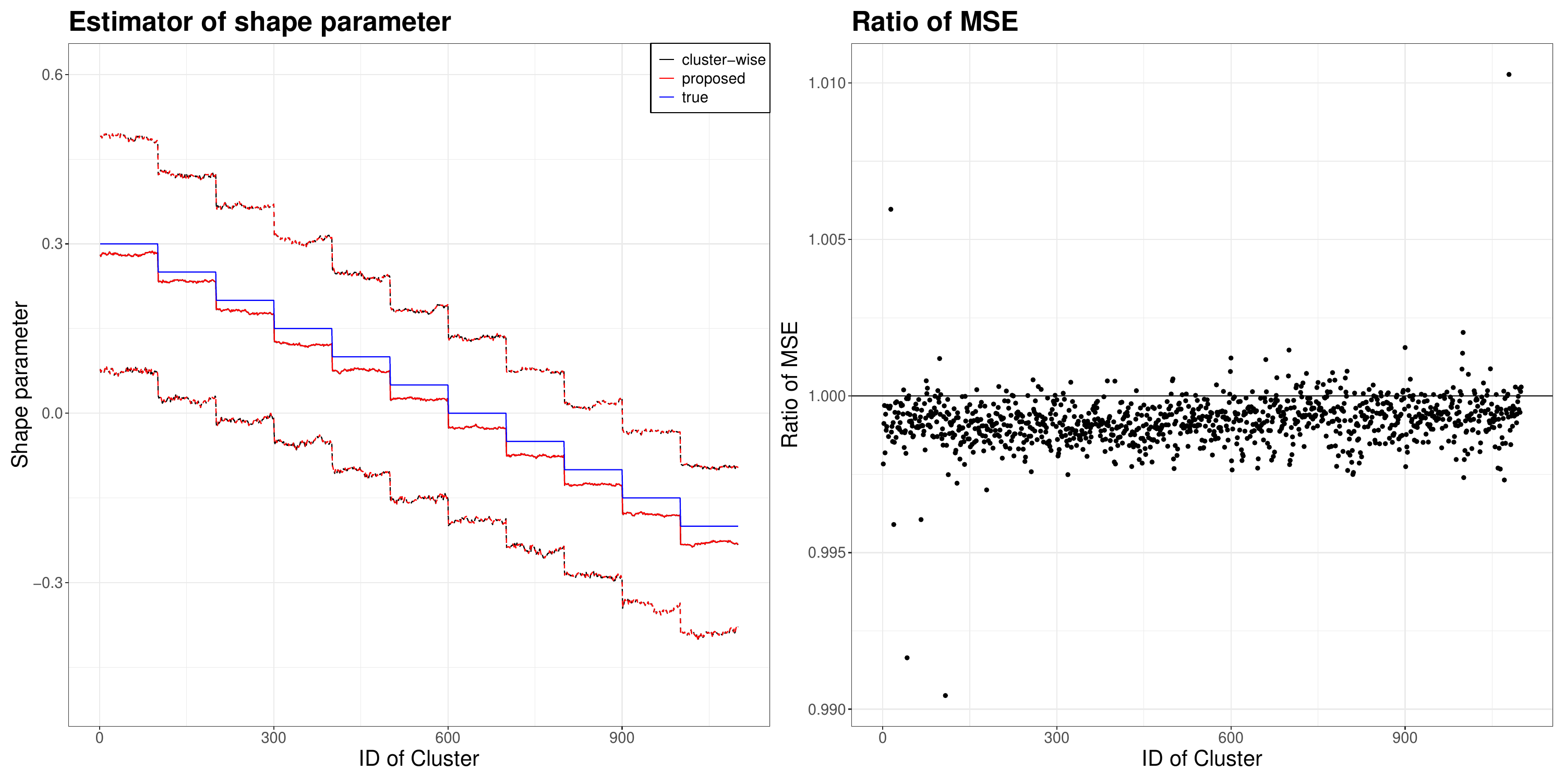}
\end{center}
\caption{Left: The cluster-wise estimator (black), the proposed estimator (red), and the true shape parameter (blue) for all clusters. 
The solid lines are medians. Dashed lines are lower/upper 2.5\% quantiles. 
Right: ratio of MSE for all clusters $j=1,\ldots,J$. 
\label{FigSection5_1}}
\end{figure}

Next, we specifically examine the return level. 
The return level of GPD for the probability $\tau\in (0,1)$ is given as 
\[
 R(\tau\mid\gamma,\sigma) =
 \left\{
 \begin{array}{cc}
 \frac{\sigma}{\gamma(\gamma+1)} \left\{\tau^{-\gamma}-1\right\}, & \gamma\not=0,\\
 -\sigma \log(\tau),& \gamma=0.
\end{array}
\right.
\]
Note that the parameter $(\gamma,\sigma)$ is defined as orthogonally transformed versions of shape and scale parameters (Chavez-Demoulin and Davison 2005).
Here, we set $\tau=1/2n$. 
Then, after constructing the 95\%-confidence interval (CI) of return level for all clusters (see Section 4.2), we calculate the coverage probability of CI and the average of ratio of length of CI of GFL over the cluster-wise estimator among 1000 Monte Carlo iterations. 
The results are presented in Figure \ref{FigSection5_2}. 
The left panel shows the coverage probability of CI whereas the ratios of lengths of CI are presented in the right panel. 
It is well known that the accuracy of confidence intervals for the return level tends not to be good in extreme value models.
Therefore the coverage probability not achieving the nominal significance level is not unexpected. 
Consequently, for almost all clusters, the coverage probabilities constructed by GFL were similar to those from cluster-wise method.
However, from right panel, for several clusters, the length of CI from GFL was shorter than that from cluster-wise method. 
This outcome might arise because GFL improves the bias (Figure \ref{FigSection5_1}). 
Furthermore, in our method, the scale parameter is estimated separately for each cluster rather than being grouped.
This helps avoid overly narrow confidence intervals, particularly for clusters with a negative shape parameter.
Consequently, although the shorter CI generally engenders reduced coverage, the bias reduction achieved by the GFL estimator allows the coverage probability to improve for several clusters.
As shown in the right panel, spike-like behaviors are observed for some clusters. 
These clusters correspond to indices around $i\times 100, i=0,\ldots, 10$, where the true shape parameter changes and grouping with neighboring clusters is less likely to occur.

\begin{figure}
\begin{center}
\includegraphics[width=150mm,height=80mm]{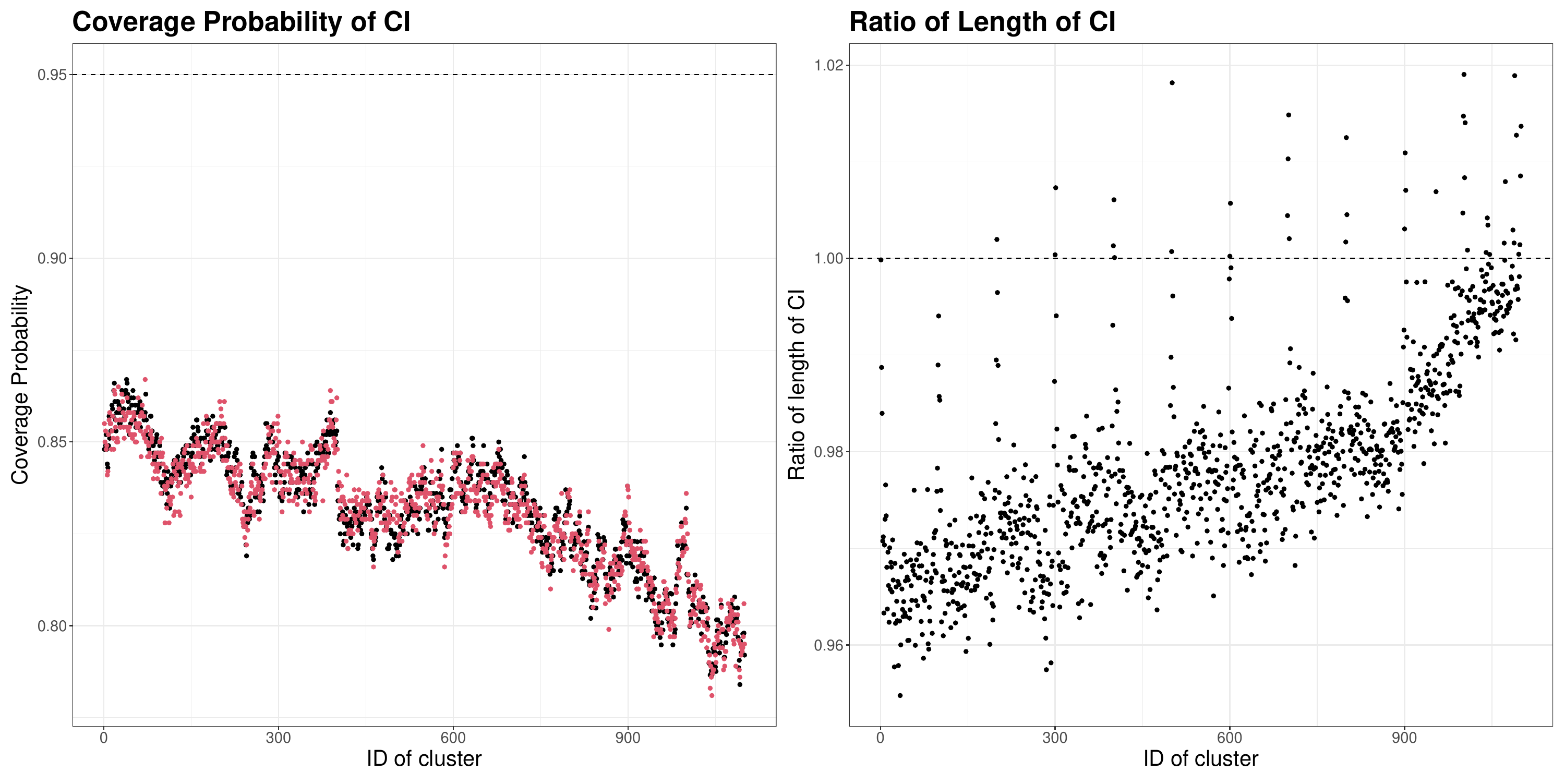}
\end{center}
\caption{ Results for 95\%-CI of return level for $j=1,\ldots,J$. 
Left: Coverage probabilities of CI from the proposed method (red) and the cluster-wise method (black). 
Right: Average of the ratio of length of CI of the proposed estimator over the cluster-wise estimator. 
\label{FigSection5_2}}
\end{figure}

Supplementary materials present additional simulation results obtained under larger sample sizes, different graph structures used in the GFL,
and alternative true models that deviate from the fully specified GPD.

\section{APPLICATION}

\subsection{Dataset}

As presented in this section, we applied the proposed graph fused lasso method to data of daily precipitation in Japan. 
Japan has over 1300 sites for recording of several meteorological data. 
The original dataset is available at the website of Japan Meteorological Agency \footnote{https://www.jma.go.jp/jma/indexe.html}. 
For this analysis, we considered 996 sites located on the main island of Japan (see Figure \ref{FigSection6_1}). Remote islands such as Okinawa were excluded because they are geographically distant from the main island and because they have distinct characteristics.
We recorded daily precipitation of 996 sites from  January 1, 2000 to December 31, 2022. 
From prior knowledge of climate information of Japan, we observed that high precipitation tended to occur during May--October, which result is consistent with expectations because these days are the rainy and typhoon season in Japan.  
Consequently, it is natural to consider that there is a seasonal effect in heavy precipitation phenomena in Japan. 
To remove such an effect, we analyzed data of May 1 to October 31 for each year at 996 sites. 
Therefore, we used 3520 observations of daily precipitation for each site. 

When analyzing heavy rainfall risk across multiple sites, nearby sites often exhibit similar meteorological patterns. 
Using this geographical structure, our goal is to improve the stability of marginal estimations at each site by incorporating spatial information. 
To this end, we applied our proposed method to rainfall data in Japan.

Before applying our method, to fit the GPD, we must choose the threshold value or number of effective sample sizes $n_j$ for all sites.
Although details are described in Section S2 of Supplementary file, here we use $n_j=116$ for all sites.

\begin{figure}
\centering
\includegraphics[width=0.45\linewidth]{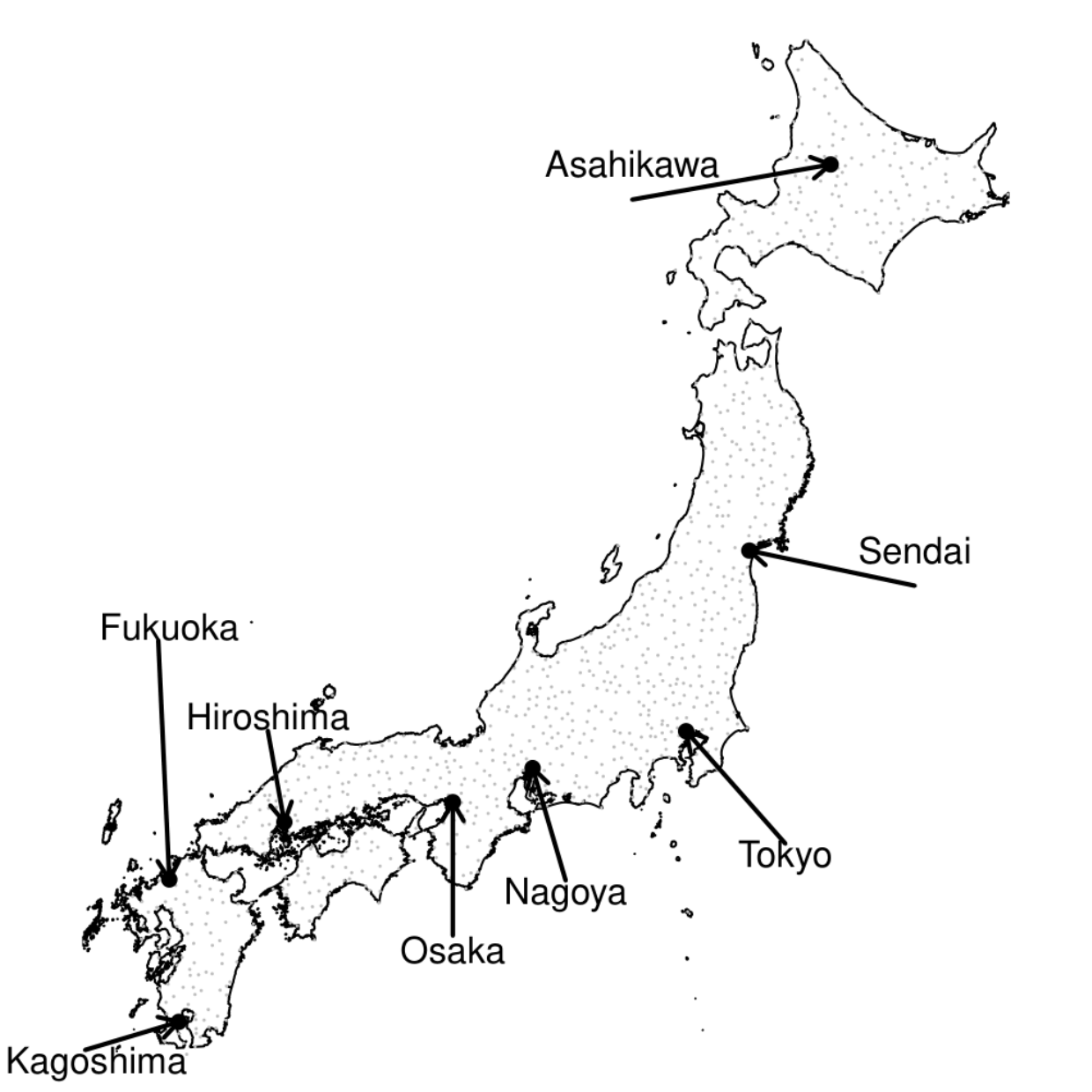}
\caption{Locations of 996 sites of the main island of Japan. \label{FigSection6_1}}
\end{figure}

\subsection{Construction of graph}

For our method, the choice of the graph in the fused lasso path is important, which is deterministic as user-specified. 
To investigate the pair of clusters having similar marginal effect, we first considered the tail dependence (see Coles et al. 1999).
For data with $j$-th and $k$-th clusters (sites), we define the asymptotic dependence as $\chi_{k,j}=\lim_{u\rightarrow 1}\chi_{k,j}(u)$, where 
\[
\chi_{k,j}(u)=P(F_k(X_{i,k})>u\mid F_j(X_{i,j})>u )=\frac{P(F_k(X_{i,k})>u, F_j(X_{i,j})>u)}{1-u}.
\]
The asymptotic dependence evaluates the event by which the uniform-transformed data of $j$-th and $k$-th exceed a given threshold $u$ simultaneously. 
If $\chi_{k,j}=0$, it is well known that data with $k$-th and $j$-th clusters are asymptotic independent, but are asymptotic dependent for $\chi_{k,j}>0$. 
In particular, $\chi_{k,j}=1$ signifies full tail dependence.
If $\chi_{k,j}$ takes the large value, then data in the $k$-th and $j$-th clusters are highly dependent in the tail region of their joint distribution, which implies that the marginal tail effects of these clusters might be similar, i.e., $\gamma_k$ might be close to $\gamma_j$.

As described herein, the edge of the graph $c_{j,k}$ was determined as 
\[
c_{j,k}= I(\chi_{j,k} > c^*),
\]
where $I$ is the indicator function and the cut-off $c^*\in(0,1)$ is a specified constant. 
That is, we added edges with a sparse penalty only to pairs having high asymptotic dependence. 
Also, we set $c^*=0.76$. 
The number of edges connecting sites by graph fused lasso was
\[
\sum_{j<k} I(\chi_{j,k}>c^*)=849. 
\]
Therefore, we can understand that two sites connected by an edge tend to experience heavy rainfall simultaneously.
Section S2 of Supplementary file describes the sensitivity of the choice of the above cut-off. 
In this setting of $c^*=0.76$, the number of single nodes for which the sites have no connection with any other site was 293. 
As single sites, they tend to be located at high elevations, in geographically isolated areas, or in places with little human activity.
The model of these 293 sites was reduced to the standard GPD approach. 
These are excluded from our proposed grouping method. 

In real data application, of course $\chi_{k,j}$ is unknown. 
The estimator of $\chi_{k,j}$ was obtained by replacing $F_j, F_k$ and $P$ with these empirical distributions and fixed $u\approx 1$. 
For this analysis, we set $u=0.98$.
The asymptotic dependence $\chi_{k,j}$ is used only to detect the pairs of clusters connecting by graph fused lasso. Its (estimated) value is not used to the weight of the graph fused lasso.

\subsection{Results}

The tuning parameter is selected using BIC. Then we obtain $\lambda=0.368$ and $K=K(\lambda) = 293$. 
The number of sites having a single node was 476. 
Actually, 293 sites had been single before application of our method. 
Consequently, 183 sites were decided by single, although grouping was considered from the method presented in Section 6.2. 
The remaining 520 sites were grouped with other sites. 
Details of results of tuning parameter selection are described in Section S2 of Supplementary file.

Estimators of shape parameters for each site are presented in Figure \ref{Figure4}. 
The left panel shows the estimated shape parameter obtained from the cluster-wise method. 
The estimator with the proposed method is plotted in the right panel. 
From this result, it is apparent that the distribution of shape parameters in the right panel is smoother than that in the left panel. 
In particular, for sites where the shape parameter was estimated as smaller than -0.17 by the cluster-wise method (marked by yellow rather than red), the proposed estimator provided slightly larger estimates. 
These larger estimates suggest that neighboring sites of those with small estimated shape parameters tend to show higher shape values and greater rainfall risk, leading the proposed estimator to yield more conservative predictions.

\begin{figure}[htbp]
\begin{center}
  \begin{minipage}[b]{0.45\linewidth}
    \centering
    \includegraphics[keepaspectratio, scale=0.3]{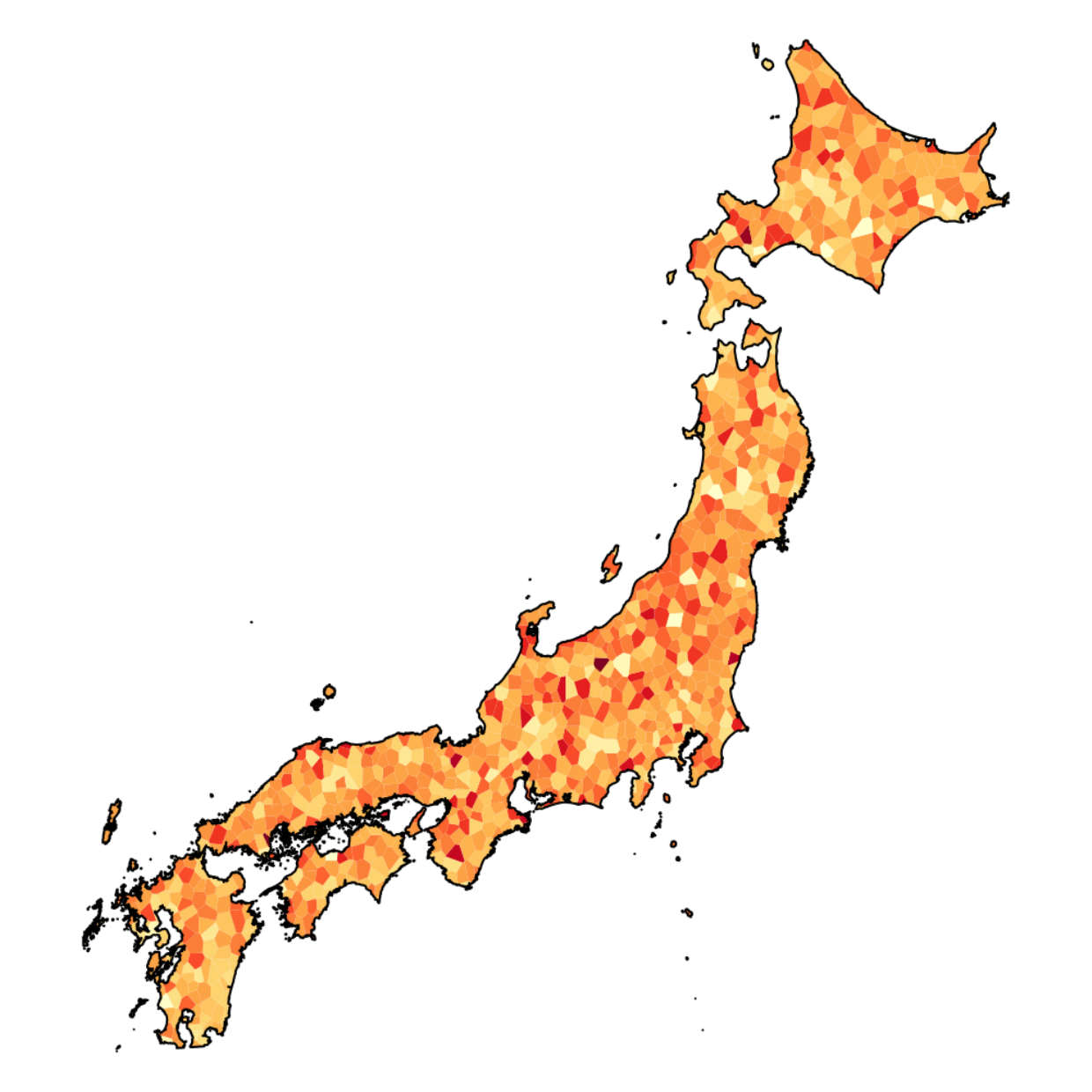}
    \vspace{0.5\baselineskip} 
    \textbf{Cluster-wise} 
  \end{minipage}
  \begin{minipage}[b]{0.45\linewidth}
    \centering
    \includegraphics[keepaspectratio, scale=0.3]{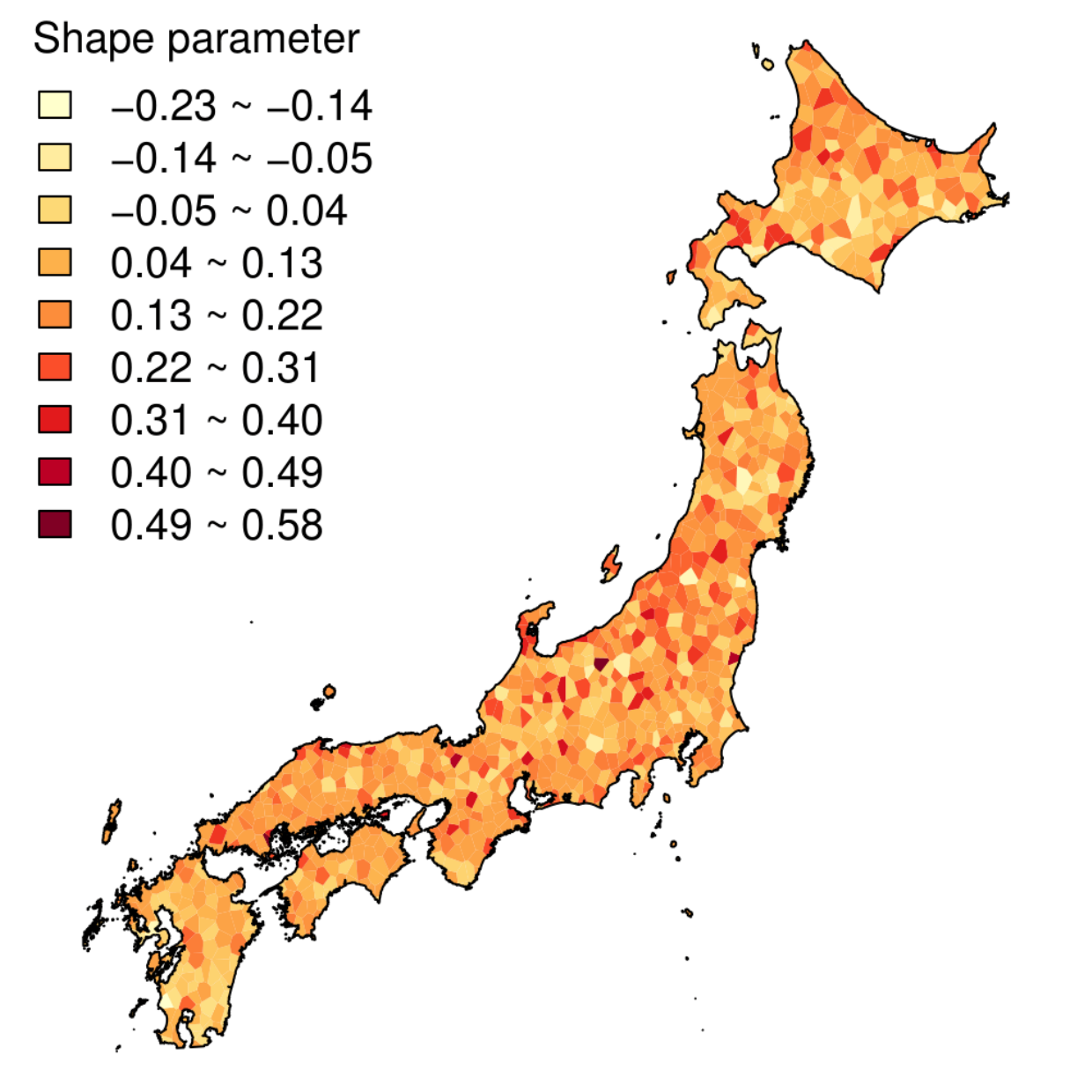}
    \vspace{0.5\baselineskip} 
    \textbf{Graph fused lasso} 
  \end{minipage}
  \end{center}
\caption{Estimator of shape parameter for each site. Left: Cluster-wise estimator. Right: Proposed estimator. \label{Figure4}}
\end{figure}

Lastly, we confirm the estimator of shape parameter, 50-year return level (RL) with 95\% confidence interval, and the group size for selected sites. 
The results are shown in Table \ref{RLtable}. 
The locations of selected sites are presented in Figure \ref{FigSection6_1}. 
At sites where the estimated $\gamma$ increased (or decreased) because of shape parameter grouping, the return levels showed a corresponding increase (or decrease) since the return level is highly sensitive to the shape parameter. 
Because of the grouping effect, the proposed method yields shorter length of  confidence intervals than the cluster-wise method does.
These results were observed for all sites. 
In Hiroshima, Osaka and Tokyo, the upper bound of the confidence interval with proposed method remains similar to that obtained from cluster-wise estimator. 
This similarity led to a more confident and accurate estimation of the return levels without increasing the risk of underestimation.
For sites at which the estimator of the shape parameter decreases (e.g. Nagoya, Sendai), both the return level and the upper bound of the confidence interval also decrease.
This decrease of both implies that, for such sites, the risk might be underestimated. 
The results should therefore be interpreted carefully in subsequent discussions.
Consequently, the proposed method is useful for summarizing information across sites by incorporating neighborhood information.
However, when assessing site-specific risks, additional information and careful interpretation are expected to be necessary.

\begin{table}
\caption{Results for eight selected sites. Estimated shape parameter, 50-year return level (RL) and 95\% confidence interval (CI) for selected sites. GS: Number of sites belonging to the same group as the selected site. \label{RLtable}}
\begin{center}
\begin{tabular}{lccc|cccc}
\hline
site& $\gamma$ & 50-year RL & 95\% CI of RL & $\gamma$ & 50-year RL & 95\% CI of RL & GS\\
\hline
&\multicolumn{3}{c}{\underline{Cluster-wise estimation}} & \multicolumn{3}{c}{\underline{Grouping estimation}}\\[2mm]
\hline
Kagoshima & -0.014 & 301.115 & (231.27, 370.96) & -0.038 & 297.175 & (273.29, 321.06) & 8\\
Fukuoka & -0.042 & 244.119 & (190.11, 298.12) & 0.099 & 321.405 & (292.60, 350.21) & 20\\
Hiroshima & 0.043 & 230.800 & (163.82, 297.78) & 0.121 & 259.042 & (242.38, 275.70) & 68\\
Osaka & 0.222 & 235.139 & (124.95, 345.32) & 0.251 & 254.950 & (189.15, 320.75) & 4\\
Nagoya & 0.317 & 311.586 & (133.11, 490.07) & 0.206 & 299.500 & (257.93, 341.07) & 13\\
Tokyo & 0.040 & 256.071 & (180.09, 332.05) & 0.040 & 256.105 & (201.91, 310.31) & 2\\
Sendai & 0.113 & 246.266 & (154.52, 338.01) & 0.138 & 222.859 & (193.40, 252.32) & 10\\
Asahikawa & 0.093 & 156.416 & (102.15, 210.68) & 0.092 & 141.053 & (112.89, 169.22) & 3\\
\hline
\end{tabular}
\end{center}
\end{table}

\begin{figure}
\begin{center}
\includegraphics[width=0.5\linewidth]{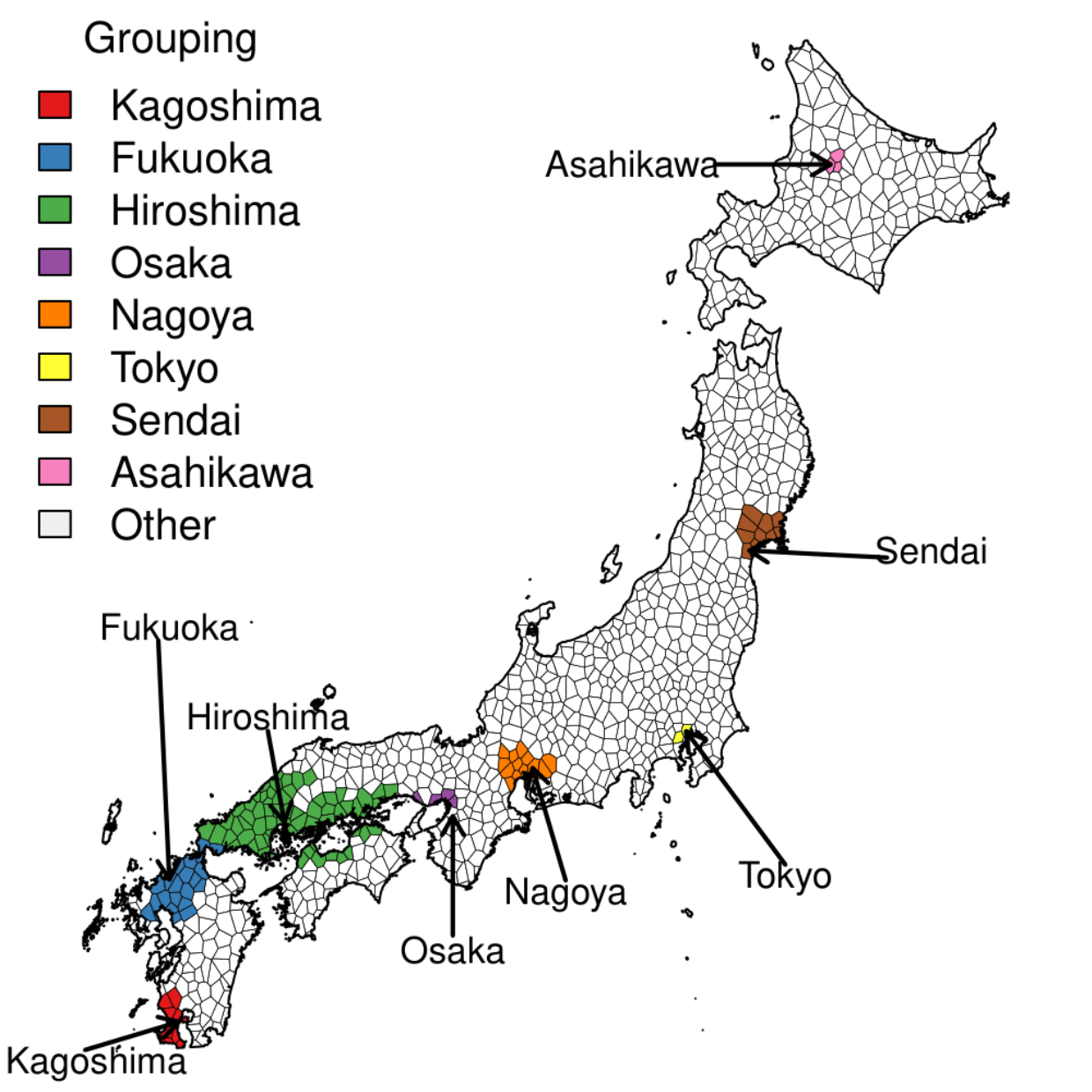}
\end{center}
\caption{Clusters grouped with selected clusters in Table 1. \label{FigSection6_4}}
\end{figure}

The last column denotes the group size (GS) of the new group to which the selected site belongs. 
For example, Kagoshima is grouped with five other (nearby) sites. 
The clusters grouped with eight clusters in Table 1 are presented in Figure \ref{FigSection6_4}. 
The selected clusters are apparently grouped with nearby clusters.
However, not all adjacent clusters are grouped together, probably because of geographic factors such as elevation and other complex climate conditions.

\section{DISCUSSION}

We have studied the structural grouping method of the marginal inference of the extreme value distribution for the clustered data. 
Specifically, we have examined the estimation of the generalized Pareto distribution (GPD) and grouping the shape parameter within the clusters having similar structure. 
As the grouping method of the shape parameter, we have used the graph fused lasso method. 
By setting the graph structure by domain knowledge, prior information or pilot study, we have estimated the shape and the scale parameter in GPD and have grouped the shape parameters simultaneously using the penalized maximum likelihood with graph fused lasso penalty. 
In our method, the graph structure between clusters is treated as user-specified arbitrary. 
This arbitrariness is sometimes inconvenient but it would be helpful in cases where the user has knowledge of the data. 
Actually, for climate data, sites that are geographically close tend to have similar features.
We have reported that such features appeared through high value of the tail dependence in Section 6. 

We now discuss further developments related to the proposed method. 
First, our proposed method is useful for the generalized extreme value distribution (GEV) of the block maxima data. 
By applying graph fused lasso penalty on the shape parameters in GEV between clusters, the estimation of parameters and grouping shape parameters could be achieved. 

A second direction is the transfer learning of the extreme value modeling for the target cluster. 
Transfer learning achieves accurate estimation of the target cluster incorporating the information of other clusters having a related structure to that of the target cluster. 
In this method, the choice of source clusters is crucially important to obtain the estimator for the target cluster. 
Using sparse method, it might be possible to estimate the parameters and the choice of source clusters simultaneously. 
Actually, Chen et al. (2015), Takada and Fujisawa (2020), and Li et al. (2022) studied transfer learning with a sparse penalty in linear regression. 
Also in our proposed method, developing transfer learning is important work.

Finally, we can consider extension to the model varying association with covariates. 
In Section 6, we applied the proposed method to daily precipitation data. 
However, the statisticalproperties of daily precipitation would change according to the time or season. 
Although we analyzed the daily precipitation data by the GPD with stationary shape and scale parameters, the time dependent shape and scale parameters would be regarded as revealing the statistical inference of the climate data. 
In this context, grouping shape parameter function as a function of time presents an interesing direction.
Similarly to a description in a report by Davison and Smith (1990), we also expect the GPD with constant shape parameter and time dependent scale parameter to be useful. 
Our proposed method could then be directly applied to group shape parameters, while the scale functions are nonparametrically estimated for each cluster.

\section*{APPENDIX}

\renewcommand{\theequation}{A\arabic{equation}}
\setcounter{equation}{0}

In this appendix, we describe the technical lemmas and these proofs for theorems. 
We denote $Y\sim GP(\gamma,\sigma)$ if the random variable $Y$ has density function 
\[
h(y\mid\gamma,\sigma) = \frac{\gamma+1}{\sigma}\left(1+\frac{\gamma(\gamma+1)y}{\sigma}\right)^{-1/\gamma-1},\ \ y>0
\]
with $\gamma>-1/2$ and $\sigma>0$. 

\begin{lemma}\label{Hessian}
If $Y\sim GP(\gamma,\sigma)$, then
\[
E\left[ \frac{\partial^2 \log h(Y\mid\gamma,\sigma)}{\partial \gamma^2} \right]
=\frac{-1}{(\gamma+1)^2},\ \ E\left[ \frac{\partial^2\log h(Y\mid\gamma,\sigma)}{\partial \gamma\partial \sigma} \right]
=0
\]
and 
\[
 E\left[ \frac{\partial^2 \log h(Y\mid\gamma,\sigma)}{\partial \sigma^2} \right]
=\frac{-\sigma^2}{2\gamma+1}.
\]
\end{lemma}

\begin{proof}[Proof of Lemma \ref{Hessian}]
The log likelihood function is 
\[
\log h(y\mid\gamma,\sigma)= -\log \sigma+\log(1+\gamma) -\left(\frac{1}{\gamma}+1\right)\log\left(1+\frac{\gamma(\gamma+1)y}{\sigma}\right).
\]
We then obtain 
\begin{eqnarray*}
\frac{\partial^2\log h(Y\mid\gamma,\sigma)}{\partial \gamma^2}&=& 
\frac{1}{(1+\gamma)^2}+
\frac{2}{\gamma^3}\log\left(1+\frac{\gamma(\gamma+1)Y}{\sigma}\right)
-\frac{1}{\gamma^2}\frac{\frac{(2\gamma+1)Y}{\sigma}}{1+ \frac{\gamma(\gamma+1)Y}{\sigma}}\\
&&-\frac{1}{\gamma^2}\frac{\frac{(2\gamma+1)Y}{\sigma}}{1+\frac{\gamma(\gamma+1)Y}{\sigma}}
+\left(\frac{1}{\gamma}+1\right)
\frac{\frac{2Y}{\sigma}\left(1+\frac{\gamma(\gamma+1)Y}{\sigma}\right)-\frac{(2\gamma+1)^2Y^2}{\sigma^2}}{\left(1+\frac{\gamma(\gamma+1)Y}{\sigma}\right)^2},\\
\frac{\partial^2 \log h(Y\mid\gamma,\sigma)}{\partial \gamma\partial\sigma}&=&
-\frac{1}{\gamma^2} \frac{\frac{-\gamma(\gamma+1)Y}{\sigma^2}}{1+\frac{\gamma(\gamma+1)Y}{\sigma}}
+\left(\frac{1}{\gamma}+1\right)\frac{\frac{-(2\gamma+1)Y}{\sigma^2}\left(1+\frac{\gamma(\gamma+1)Y}{\sigma}\right)-\frac{(2\gamma+1)Y}{\sigma}\frac{\gamma(\gamma+1)Y}{\sigma^2}}{\left(1+\frac{\gamma(\gamma+1)Y}{\sigma}\right)^2}\\
&=&
\frac{1}{\gamma^2} \frac{\frac{\gamma(\gamma+1)Y}{\sigma^2}}{1+\frac{\gamma(\gamma+1)Y}{\sigma}}
-\left(\frac{1}{\gamma}+1\right)\frac{\frac{(2\gamma+1)Y}{\sigma^2}}{\left(1+\frac{\gamma(\gamma+1)Y}{\sigma}\right)^2},
\end{eqnarray*}
and 
\[
\frac{\partial^2 \log h(Y\mid\gamma,\sigma)}{\partial \sigma^2}
=
\left(\frac{1}{\gamma}+1\right)\frac{\sigma^2\frac{\gamma(\gamma+1)Y}{\sigma}}{\left(1+\frac{\gamma(\gamma+1)Y}{\sigma}\right)^2}.
\]
To calculate the expectation of the above, we use the results as
\begin{eqnarray*}
&&E\left[\log\left(1+\frac{\gamma(\gamma+1)Y}{\sigma}\right)\right]=\gamma,\ \ 
E\left[\left(1+\frac{\gamma(\gamma+1)Y}{\sigma}\right)^{-1}\right]=
\frac{1}{(\gamma+1)},\\
&&E\left[\left(1+\frac{\gamma(\gamma+1)Y}{\sigma}\right)^{-2}\right]
=
\frac{1}{2\gamma+1},\ \ 
E\left[\frac{\frac{(2\gamma+1)Y}{\sigma}}{1+\frac{\gamma(\gamma+1)Y}{\sigma}}\right]
=
   \frac{2\gamma+1}{(\gamma+1)^2},\\
 &&  E\left[\frac{\frac{(2\gamma+1)Y}{\sigma}}{\left(1+\frac{\gamma(\gamma+1)Y}{\sigma}\right)^2}\right]
=
   \frac{1}{(\gamma+1)^2},
\end{eqnarray*}
and
\[
E\left[\frac{\frac{(2\gamma+1)^2Y^2}{\sigma^2}}{\left(1+\frac{\gamma(\gamma+1)Y}{\sigma}\right)^2}\right]
=
 \frac{2(2\gamma+1)}{(\gamma+1)^3}.
\]
Then, we obtain 
\begin{eqnarray*}
E\left[\frac{\partial^2 \log h(Y\mid\gamma,\sigma)}{\partial \gamma^2}\right]&=& -\frac{1}{(\gamma+1)^2},\\
E\left[\frac{\partial^2 \log h(Y\mid\gamma,\sigma)}{\partial \gamma\partial \sigma}\right]&=&0
\end{eqnarray*}
and 
\[
E\left[\frac{\partial^2 \log h(Y\mid\gamma,\sigma)}{\partial \sigma^2}\right] = -\frac{\sigma^2}{2\gamma+1}.
\]
\end{proof}

Before describing proof of Theorem \ref{Oracle}, we state the second-order condition of EVT of each $F_j\ (j=1,\ldots,J)$. 
For $\tau\in(0,1)$, let $H^{-1}(\tau\mid\gamma)$ is the inverse function of $H(y\mid\gamma)$, i.e., $H^{-1}(\tau\mid\gamma)=\{t : H(t\mid\gamma)=\tau\}$. 
We further define
\[
Q(x\mid\gamma,\rho)=H(x\mid\gamma)^{1+\gamma}\frac{1}{\rho}\left\{\frac{\{H^{-1}(x\mid\gamma)\}^{\gamma+\rho}-1}{\gamma+\rho}-\frac{\{H^{-1}(x\mid\gamma)\}^{\gamma}-1}{\gamma}\right\}
\]
with parameter $\gamma\in\mathbb{R}$ and $\rho \leq 0$ for $x>0$. 
We can define $Q(x\mid 0, \rho)=\lim_{\gamma\rightarrow 0} Q(x\mid\gamma, \rho)$ for $\rho<0$, $Q(x\mid\gamma, 0)=\lim_{\rho\rightarrow 0}Q(x\mid\gamma, \rho)$ for $\gamma\not=0$ and $Q(x\mid0, 0)=\lim_{\gamma,\rho\rightarrow 0} Q(x\mid \gamma, \rho)$. 
Then, the condition (\ref{SecondOrderSimple}) is detailed here. 
For $F_j\in{\cal D}(G_{\gamma_j})$, we now assume that, for the sequence of threshold $w_{j}$, there exists a parameter $\rho_j$, a sequence $\sigma_{j,w_j}$, and an auxiliary function $\alpha_j(w_j)$ satisfying $\alpha_j(w_j)\rightarrow 0$ as $w_j\rightarrow x_j^*:=\sup\{x: F_j(x)<1\}$ such that for any $x>0$,
\begin{eqnarray}
\lim_{w_j\rightarrow x_j^*}\left|\frac{\frac{1-F_j(w_j+x)}{1-F_j(w_j)}- H((\gamma_j+1)x/\sigma_{j,w_j} \mid \gamma_j)}{\alpha_j(w_j)}- Q\left(\frac{(\gamma_j+1)x}{\sigma_{j,w_j}}\middle| \gamma_j,\rho_j\right)\right| \rightarrow 0. \label{SecondOrder}
\end{eqnarray}
The condition (\ref{SecondOrder}) is addressed specifically in Theorem 2.3.8 of reported work by de Haan and Ferreira (2006).
For $\gamma_j>0$, $\alpha_j(w_j)$ satisfies $\alpha_j(zw_j)/\alpha_j(w_j) \rightarrow z^{\rho_j/\gamma_j}$ as $w_j\rightarrow\infty$. 

Let $f_{j, w_j}(y_{ij})$ be the density function of $Y_{ij}=X_{ij}-w_j$ under $X_{ij}>w_j$, for $X_{ij}\sim F_j$, where $F_j\in{\cal D}(G_{\gamma_j})$. 
According to (\ref{SecondOrder}), the density function $f_{j,w_j}(y)=d P(Y_{ij}<y|Y_{ij}>0)/dy$ satisfies the following.
\begin{eqnarray}
f_{j,w_j}(y) = h\left(\frac{y}{\sigma_{j,w_j}}\middle| \gamma_j\right) -\alpha_j(w_j) Q^{\prime}\left(\frac{(\gamma+1)y}{\sigma_{j,w_j}}\middle|\gamma_j,\rho_j\right)(1+o(1)) \label{SecondOrderDensity}
\end{eqnarray}
for $y=x-w_j>0$ as $w_j\rightarrow x_j^*$, where 
\[
Q^{\prime}((\gamma+1)y/\sigma_{j,w_j}\mid\gamma_j,\rho_j)= \frac{d}{dy}Q((\gamma+1)y/\sigma_{j,w_j}\mid\gamma_j,\rho_j).
\]
In fact, $h(y/\sigma_{j,w_j}\mid\gamma_j)= d H((\gamma+1)y/\sigma_{j,w_j}\mid\gamma_j) / dy$ from (\ref{rescaleGP}). 

This detailed notation of second order condition of EVT helps to understand the proof of Theorem \ref{Oracle}.

\begin{proof}[Proof of Theorem \ref{Oracle}]
Consequently, we have
\[
\ell_{{\cal A}}(\gamma,\sigma_1,\ldots,\sigma_J)=\sum_{j=1}^J \sum_{i=1}^n I(Y_{ij}>0)\log h(Y_{ij}\mid\gamma, \sigma_j).  
\]
Then, the maximizer of $\ell_{{\cal A}}$ is $(\hat{\gamma}, \hat{\sigma}_1,\ldots,\hat{\sigma}_J)$. 
Next we calculate the gradient and Hessian of $\ell$. 
First, we obtain 
\begin{eqnarray*}
\frac{\partial \ell_{\cal A}(\gamma,\sigma_1,\ldots,\sigma_J)}{\partial \gamma}
=\sum_{j=1}^J \sum_{i=1}^nI(Y_{ij}>0)\frac{\partial\log h(Y_{ij}\mid\gamma, \sigma_j)}{\partial \gamma}.
\end{eqnarray*}
From (\ref{SecondOrderDensity}), 
\begin{eqnarray*}
E\left[\frac{\partial \log h(y\mid\gamma, \sigma_j)}{\partial \gamma}\right]
&=& 
\int  \frac{\partial\log h(y\mid\gamma, \sigma_j)}{\partial \gamma} h(y\mid\gamma,\sigma_j) dy\\
&&+\alpha_j(w_j)\int \frac{\partial\log h(y\mid\gamma, \sigma_j)}{\partial \gamma}  Q^\prime((\gamma+1)y/\sigma_j\mid\gamma,\rho_j)dy(1+o(1)).
\end{eqnarray*}
By the definition of derivative of log-likelihood function, we have 
\[
\int  \frac{\partial\log h(y\mid\gamma, \sigma_j)}{\partial \gamma} h(y\mid\gamma,\sigma_j) dy=0.
\]
In addition, by the straightforward calculation, it is readily apparent that 
\[
\left|\int \frac{\partial\log h(y\mid\gamma, \sigma_j)}{\partial \gamma}  Q^\prime((\gamma+1)y/\sigma_j\mid\gamma,\rho_j)dy\right| <\infty.
\]
Therefore, we have 
\[
E\left[\frac{\log h(Y_{ij}\mid\gamma, \sigma_j)}{\partial \gamma}\right]
= o(n_{{\cal A}}^{-1/2}),
\]
which indicates that 
\begin{eqnarray}
\frac{1}{n_{{\cal A}}^{1/2}}\left|E\left[\frac{\partial \ell_{\cal A}(\gamma,\sigma_1,\ldots,\sigma_J)}{\partial \gamma}\right]\right|\leq O(n_{{\cal A}}^{1/2} \max_j \alpha_j(w_j) )=o(1). \label{ExGamma}
\end{eqnarray}
Next, from Lemma \ref{Hessian}, it is apparent that 
\[
V\left[\frac{\partial \log h(Y_{ij}\mid\gamma, \sigma_j)}{\partial \gamma} \right]= \frac{n_{{\cal A}}}{(\gamma+1)^2}(1+o(1)).
\]
Therefore, we obtain 
\begin{eqnarray*}
V\left[\frac{1}{n_{{\cal A}}^{1/2}}\frac{\partial \ell_{\cal A}(\gamma,\sigma_1,\ldots,\sigma_J)}{\partial \gamma}\right] =  \frac{1}{(\gamma+1)^2}(1+o(1)).
\end{eqnarray*}
We see from central limit theorem and the Cramer--Wold device that
\[
\left\{n_1^{-1/2}\sum_{i=1}^nI(Y_{i1}>0)\frac{\partial \log h(Y_{i1}\mid\gamma,\sigma_1)}{\partial \gamma},\ldots, n_J^{-1/2}\sum_{i=1}^nI(Y_{iJ}>0)\frac{\partial \log h(Y_{iJ}|\gamma,\sigma_J)}{\partial \gamma} \right\}
\]
is asymptotically distributed as normal with covariance matrix $(\gamma+1)^2 I_J$, where $I_J$ is the $J$-identity matrix. 
This outcome implies that
\[
\frac{\partial \ell_{\cal A}(\gamma,\sigma_1,\ldots,\sigma_J)}{\partial \gamma}
=\frac{1}{n_{{\cal A}}^{1/2}}\sum_{j=1}^J n_j^{1/2} \frac{1}{n_j^{1/2}}\sum_{i=1}^nI(Y_{ij}>0)\frac{\partial \log h(Y_{ij}\mid\gamma,\sigma_j)}{\partial \gamma}
\]
is asymptotically distributed as normal with mean 0 and variance $(\gamma+1)^2$. 
Consequently, 
\[
n_{{\cal A}}^{1/2}(\hat{\gamma}-\gamma) \xrightarrow{D} N(0, (\gamma+1)^2)
\]
holds. 
Next, for $j=1,\ldots,J$, we have
\[
\frac{\partial^2 \ell_{\cal A}(\gamma,\sigma_1,\ldots,\sigma_J)}{\partial \gamma \partial \sigma_j}
=\sum_{i=1}^nI(Y_{ij}>0) \frac{\partial^2 \log h(Y_{ij}\mid\gamma,\sigma_j)}{\partial\gamma\partial\sigma_j}
\]
and we have 
\[
\int \frac{\partial^2 \log h(y\mid\gamma,\sigma_j)}{\partial\gamma\partial\sigma_j} h(y|\gamma,\sigma_j)dy=0
\]
from Lemma \ref{Hessian}. 
Tedious but straightforward calculation implies 
\[
\frac{1}{\sigma_j}\left|\int \frac{\partial^2 \log h(y\mid\gamma,\sigma_j)}{\partial\gamma\partial\sigma_j} Q^\prime((\gamma+1)y/\sigma_j\mid\gamma,\rho_j)dy\right| <\infty.
\]
Consequently, the covariance between $\hat{\gamma}$ and $\hat{\sigma}_j/\sigma_j$ converges to zero. 
Similarly to the calculations above, we obtain 
\[
n_j^{1/2}\left(\frac{\hat{\sigma}_j}{\sigma_j} - 1\right)  \xrightarrow{D} N(0, 2\gamma + 1),\ \ j=1,\ldots, J. 
\]
\end{proof}

\begin{proof}[Proof of Theorem \ref{MainTheorem}]

Let 
\[
(\gamma_{0,{\cal B}_{k,\ell}}, \sigma_{0,j}, j\in{\cal B}_{k,\ell}) = \underset{\gamma,\sigma_j,j\in{\cal B}_{k,\ell}}{\argmin}\ \ \sum_{j\in{\cal B}_{k,\ell}}n_jE\left[-\log h(Y_{ij}|\gamma, \sigma_j) \right] ,\ \ \ell=1,\ldots,L_k, k=1,\ldots,K.
\]
We then define
\begin{eqnarray*}
{\cal M}
=\left\{
(\vec{\gamma},\vec{\sigma}) : \sqrt{n_{{\cal B}_\ell}}|\gamma_j-\gamma_{0,{\cal B}_{k,\ell}}|< C, \sqrt{n_j}\left|\frac{\sigma_{j}}{\sigma_{0,j}}-1\right|<C, j\in{\cal B}_{k,\ell}, \ell=1,\ldots,L_k,k=1,\ldots,K
\right\}
\end{eqnarray*}
and 
\begin{eqnarray*}
{\cal M}_n&=&{\cal M}_n(t_n)\\
&=&
{\cal M}\cap \left\{
(\vec{\gamma},\vec{\sigma}) : |\gamma_j-\hat{\gamma}_{{\cal B}_{k,\ell}}|< t_n, \left|\frac{\sigma_{j}}{\hat{\sigma}_{{\cal B}_{k,\ell},j}}-1\right|<t_n, j\in{\cal B}_{k,\ell}, \ell=1,\ldots,L_k, k=1,\ldots,K
\right\}
\end{eqnarray*}
for sequence $t_n>0$.

Because the Hessian matrix of $\ell_F$ is positive definite, $\ell_F$ has a unique minimizer. 
Therefore, if we show 
\begin{eqnarray*}
\ell_F(\vec{\gamma},\vec{\sigma}) \geq \ell_F(\hat{\vec{\gamma}}_{{\rm oracle}}, \hat{\vec{\sigma}}_{{\rm oracle}})\ \ {\rm for\ any}\  (\vec{\gamma},\vec{\sigma})\in {\cal M}_n, 
\end{eqnarray*}
 $(\hat{\vec{\gamma}}_{{\rm oracle}}, \hat{\vec{\sigma}}_{{\rm oracle}})$ is a strictly local minimizer of $\ell_F$.

 For $(\vec{\gamma},\vec{\sigma})\in{\cal M}$, $\vec{\gamma}^*=(\gamma_1^*,\ldots,\gamma_J^*)$ denotes the orthogonal projection on to the oracle space: $\Gamma_{{\rm oracle}} = \{\vec{\gamma}: \gamma_i=\gamma_j, i,j\in{\cal B}_{k,\ell}, \ell=1,\ldots,L_k, k=1,\ldots,K\}$. 
Also, $(\vec{\gamma}^*, \vec{\sigma})\in{\cal M}$ is satisfied.
To show (\ref{PurposeProof}), we investigate the following two inequalities:
\begin{eqnarray}
\ell_F(\vec{\gamma},\vec{\sigma}) \geq \ell_F(\vec{\gamma}^*,\vec{\sigma}) \geq \ell_F(\hat{\vec{\gamma}}_{{\rm oracle}}, \hat{\vec{\sigma}}_{{\rm oracle}}). \label{PurposeProof}
\end{eqnarray}

We first show
\begin{eqnarray}
\ell_F(\vec{\gamma}^*,\vec{\sigma}) \geq \ell_F(\hat{\vec{\gamma}}_{{\rm oracle}}, \hat{\vec{\sigma}}_{{\rm oracle}}) \label{oracleStar}
\end{eqnarray}
for any $(\vec{\gamma}^*, \vec{\sigma})\in {\cal M}$. 
From Theorem \ref{Oracle}, it is readily apparent that $(\hat{\vec{\gamma}}_{{\rm oracle}},\hat{\vec{\sigma}}_{{\rm oracle}})\in{\cal M}$ with probability tends to one. 
Therefore, the penalty term of $\ell_F(\hat{\vec{\gamma}}_{{\rm oracle}}, \hat{\vec{\sigma}}_{{\rm oracle}})$ is 
\begin{eqnarray}
n\lambda \sum_{i,j} c_{i,j}v_{j,k}|\hat{\gamma}_{i}-\hat{\gamma}_j| = 
n\lambda \sum_{i,j\in{\cal B}^*} c_{i,j}v_{j,k}|\hat{\gamma}_{i}-\hat{\gamma}_j|. \label{penOracle}
\end{eqnarray}
For $i,j\in{\cal B}^*$, there exist $k, h$ with $k\not=h$ such that $\hat{\gamma}_i=\hat{\gamma}_{{\cal A}_k}$ and $\hat{\gamma}_j=\hat{\gamma}_{{\cal A}_h}$. 
For any $j=1,\ldots,J$, the cluster-wise estimator satisfies $\tilde{\gamma}_j= \gamma_j+ O(n_j^{-1/2})$.
 Therefore, under (A3), for $i,j\in{\cal B}^*$, we have $|\tilde{\gamma}_i-\tilde{\gamma}_j|=\delta + O(\max_j n_j^{-1/2})$ as $n\rightarrow\infty$, which implies that, for $i,j\in {\cal B}^*$, 
\[
v_{j,k} =p_\lambda(|\tilde{\gamma}_i-\hat{\gamma}_j|) \leq p_\lambda(a\lambda) =0
\]
because $|\tilde{\gamma}_i-\tilde{\gamma}_j|>\delta/2 > a\lambda$ for sufficiently large $n$. 
Consequently, (\ref{penOracle}) converges to zero as $n\rightarrow\infty$. 
By a similar argument, because $(\vec{\gamma}^*,\vec{\sigma})\in{\cal M}$, the penalty term of $\ell_F(\vec{\gamma}^*,\vec{\sigma})$ also converges to zero as $n\rightarrow\infty$. 
By the definition of $(\hat{\vec{\gamma}}_{{\rm oracle}},\hat{\vec{\sigma}}_{{\rm oracle}})$, we have 
\[
\ell(\vec{\gamma}^*, \vec{\sigma}) \geq \ell(\hat{\vec{\gamma}}_{{\rm oracle}},\hat{\vec{\sigma}}_{{\rm oracle}})
\]
for any $(\vec{\gamma}, \vec{\sigma})\in\mathbb{R}^J\times\mathbb{R}^J$. 
Consequently, with probability tends to one, (\ref{oracleStar}) holds.

We next present the first inequality of (\ref{PurposeProof}), i.e., 
\begin{eqnarray}
\ell_F(\vec{\gamma},\vec{\sigma}) - \ell_F(\vec{\gamma}^*,\vec{\sigma}) \geq 0.
\end{eqnarray}
Similarly to the inequality presented above, the penalty term of $\ell_F(\vec{\gamma}^*,\vec{\sigma})$ becomes zero with probability tending to one. 
Next we consider the penalty term of $\ell_F(\vec{\gamma},\vec{\sigma})$. 
For $j,k\in{\cal B}^*$, with probability tends to one, we have $|\tilde{\gamma}_j-\tilde{\gamma}_k|>\delta/2$, which implies that $v_{j,k}=\rho_\lambda(|\tilde{\gamma}_j-\tilde{\gamma}_k|)\leq p_\lambda(a\lambda)=0$. 
For $i,j\in{\cal B}_{k,\ell}$, we have $|\tilde{\gamma}_j-\tilde{\gamma}_k| \leq O(\max\{n_j^{-1/2},n_k^{-1/2}\})<\lambda$. Therefore,
\[
v_{j,k}=p_\lambda(|\tilde{\gamma}_i-\tilde{\gamma}_j|)\geq |\rho_\lambda(\lambda)|> a_0.
\]
This result implies that, for sufficiently large $n$, the penalty term of  $\ell_F(\vec{\gamma},\vec{\sigma})$ is 
\begin{eqnarray*}
&&n\lambda  \sum_{k=1}^K\sum_{\ell=1}^{L_k}\sum_{i,j\in{\cal B}_{k,\ell}}v_{j,k}c_{i,j}|\gamma_i-\gamma_j| +\lambda  \sum_{j,k\in{\cal A}^*} v_{j,k}c_{j,k}|\gamma_j-\gamma_k|\\
&&
=n\lambda  \sum_{k=1}^K\sum_{\ell=1}^{L_k}\sum_{i,j\in{\cal B}_{k,\ell}}v_{j,k}c_{i,j}|\gamma_i-\gamma_j| \\
&&\geq n \lambda a_0\sum_{k=1}^K\sum_{\ell=1}^{L_k}\sum_{i,j\in{\cal B}_{k,\ell}}c_{i,j}|\gamma_i-\gamma_j| \\
&&\geq  n\lambda a_0  \sum_{k=1}^K\sum_{\ell=1}^{L_k}\max_{i,j\in{\cal B}_{k,\ell}} |\gamma_i-\gamma_j|.
\end{eqnarray*}

For $\vec{\gamma}^*\in\Gamma_{{\rm oracle}}$, we write $\gamma_j^*=\gamma_{{\cal B}_{k,\ell}}^*$ for $j\in{\cal B}_{k,\ell}$ for $\ell
=1,\ldots,L_k, k=1,\ldots,K$. 
Then we consider the log-likelihood part as
\[
\ell(\vec{\gamma},\vec{\sigma})- 
\ell(\vec{\gamma}^*,\vec{\sigma})
=\sum_{k=1}^K \sum_{\ell=1}^{L_k} \sum_{j\in{\cal B}_{k,\ell}}\sum_{i=1}^nI(Y_{ij}>0)\left\{\log h(Y_{ij}\mid\gamma_j,\sigma_j)-\log h(Y_{ij}\mid\gamma^*_{{\cal B}_{k,\ell}},\sigma_{j})\right\}.
\]

For any $(k,\ell)$, by the Taylor expansion, we obtain
\begin{eqnarray}
&&\sum_{j\in{\cal B}_{k,\ell}}\sum_{i=1}^nI(Y_{ij}>0)\left\{\log h(Y_{ij}\mid\gamma_j,\sigma_j)-\log h(Y_{ij}\mid\gamma^*_{{\cal B}_{k,\ell}},\sigma_{j})\right\}\nonumber\\
&&=
\sum_{j\in{\cal B}_{k,\ell}} (\gamma_j-\gamma_{{\cal B}_{\ell,k}}^*) \sum_{i=1}^nI(Y_{ij}>0)\frac{\partial \log h(Y_{ij}\mid\gamma_j^\dagger,\sigma_j)}{\partial \gamma_j}, \label{TaylorGamma}
\end{eqnarray}
where $\gamma_j^\dagger\in\mathbb{R}$ satisfies $|\gamma_j-\gamma^\dagger_j|< |\gamma_j-\gamma^*_{{\cal B}_{k,\ell}}|$ for $j\in{\cal B}_{k,\ell}$. 
We further take the Taylor expansion to $\partial \log h(Y_{i,j}|\gamma_j^\dagger,\sigma_j)/\partial \gamma_j$ around $(\gamma_j,\sigma_j)=(\gamma_{0,j},\sigma_{0,j})$ as 
\begin{eqnarray*}
\frac{\partial \log h(Y_{ij}\mid\gamma_j^\dagger,\sigma_j)}{\partial \gamma_j}
&=&
\frac{\partial \log h(Y_{ij}\mid\gamma_{0,j},\sigma_{0,j})}{\partial \gamma_j}\\
&+&\left\{\frac{\partial^2 \log h(Y_{ij}\mid\gamma_{0,j},\sigma_{0,j})}{\partial^2 \gamma_j}(\gamma_j^\dagger-\gamma_{0,j})
+
\frac{\partial^2 \log h(Y_{ij}\mid\gamma_{0,j},\sigma_{0,j})}{\partial \gamma_j\partial \sigma_j}(\sigma_j-\sigma_{0,j})
 \right\}(1+o_P(1)).
\end{eqnarray*}
From Lemma \ref{Hessian}, we have 
\[
\frac{1}{\sqrt{n_j}}\sum_{i=1}^nI(Y_{ij}>0)\frac{\partial \log h(Y_{ij}\mid\gamma_{0,j},\sigma_{0,j})}{\partial \gamma_j} =O_P(1).
\]
Because 
\[
E\left[\frac{\partial^2 \log h(Y_{ij}\mid\gamma_{0,j},\sigma_{0,j})}{\partial \gamma_j\partial \sigma_j}\right]=0
\]
and $|\sigma_j/\sigma_{0,j}-1|< O(n_j^{-1/2})$, 
we have
\[
\left|\frac{1}{\sqrt{n_j}}\sum_{i=1}^nI(Y_{ij}>0)\frac{\partial^2 \log h(Y_{ij}\mid\gamma_{0,j},\sigma_{0,j})}{\partial \gamma_j\partial \sigma_j}(\sigma_j-\sigma_{0,j})\right| = o_P(1). 
\]
Consequently, there exists the constant $C^*>0$ such that for any $j\in{\cal B}_{k,\ell}$ with $\ell=1,\ldots,L_k, k=1,\ldots,K$, 
\[
\left|\sum_{i=1}^nI(Y_{ij}>0) \frac{\partial \log h(Y_{ij}\mid\gamma_j^\dagger,\sigma_j)}{\partial \gamma_j}\right| \leq 
C^*n_j^{-1/2}|\gamma_j^\dagger - \gamma_{0,j}|.
\]
with probability tending to one. 
In fact, $\gamma_{0,j}=\gamma_{0,{\cal B}_{k,\ell}}$ for $j\in{\cal B}_{k,\ell}$. 
Then, by the definition of $\gamma_j^*$, we obtain 
$|\gamma_j^*-\gamma_{j}| \leq |\gamma_j^*-\hat{\gamma}_{{\cal B}_{k,\ell}}| \leq t_n$. 
From this and the definition of $\gamma_j^\dagger$, it is readily apparent that for any $j\in{\cal B}_{k,\ell}$,
\begin{eqnarray}
|\gamma_j^\dagger-\gamma_{0_j}|\leq |\gamma_j^\dagger-\gamma_j^*|+|\gamma_j^*-\gamma_{0,j}|\leq  |\gamma_j-\gamma_j^*| + |\gamma_j-\gamma_{0,j}|\leq t_n +  C n_{{\cal B}_{k,\ell}}^{-1/2}. \label{Bound}
\end{eqnarray}
Therefore, we obtain that for $j\in{\cal B}_{\ell,k}$, 
\[
\left|\sum_{i=1}^nI(Y_{ij}>0) \frac{\partial \log h(Y_{ij}\mid\gamma_j^\dagger,\sigma_j)}{\partial \gamma_j}\right| \leq C_1 n_j^{1/2}t_n + C_2 (n_j/n_{{\cal B}_{k,\ell}})^{1/2}
\]
for some constant $C_1, C_2>0$, with probability tending to one.
Because $\gamma_{{\cal B}_{k,\ell}}^*$ is the orthogonal projection of $\gamma_j, j\in{\cal B}_{k,\ell}$ on to $\Gamma_{{\rm oracle}}$, we can express 
\[
\gamma_{{\cal B}_{k,\ell}}^* = \frac{1}{|{\cal B}_{k,\ell}|}\sum_{j\in{\cal B}_{k,\ell}} \gamma_j.
\]
Therefore, we obtain 
\begin{eqnarray}
\sum_{j\in{\cal B}_{k,\ell}} (\gamma_j-\gamma_{{\cal B}_{\ell,k}}^*)
\leq
 \frac{1}{|{\cal B}_{k,\ell}|}\sum_{i,j\in{\cal B}_{k,\ell}}|\gamma_i-\gamma_j| \leq |{\cal B}_{k,\ell}| \max_{i,j\in{\cal B}_{k,\ell}}|\gamma_i-\gamma_j|. \label{OrthogonalIneq}
\end{eqnarray}

From (\ref{TaylorGamma}), (\ref{Bound}), and (\ref{OrthogonalIneq}), we obtain 
\begin{eqnarray*}
|\ell(\vec{\gamma},\vec{\sigma})-  \ell(\vec{\gamma}^*,\vec{\sigma})|
\leq  \sum_{k=1}^K\sum_{\ell=1}^{L_k} \left\{C_1 t_n \max_{j\in{\cal B}_{k,\ell}} n_j^{1/2} + C_2 \frac{\max_{j\in{\cal B}_{k,\ell}} n_j^{1/2}}{n_{{\cal B}_{k,\ell}}^{1/2}}\right\}  |{\cal B}_{k,\ell}| \max_{i,j\in{\cal B}_{k,\ell}}|\gamma_i-\gamma_j|.
\end{eqnarray*}
Therefore, we have 
\begin{eqnarray*}
&&\ell_F(\vec{\gamma}, \vec{\sigma})-\ell_F(\vec{\gamma}^*, \vec{\sigma})\\
&&\geq
n\lambda a_0  \sum_{k=1}^K\sum_{\ell=1}^{L_k}
\left\{
1- 
 \frac{ |{\cal B}_{k,\ell}|}{n\lambda a_0}\left\{C_1 t_n \max_{j\in{\cal B}_{k,\ell}} n_j^{1/2} + C_2 \frac{\max_{j\in{\cal B}_{k,\ell}} n_j^{1/2}}{n_{{\cal B}_{k,\ell}}^{1/2}}\right\} 
\right\}\max_{i,j\in{\cal B}_{k,\ell}} |\gamma_i-\gamma_j|.
\end{eqnarray*}
From the definition of $n_j$, it is readily apparent that 
\[
 \frac{\max_{j\in{\cal B}_{k,\ell}} n_j^{1/2}}{n_{{\cal B}_{k,\ell}}^{1/2}} \rightarrow 0
\]
as $n\rightarrow\infty$. 
Next, by taking $t_n=(\max_{1\leq j\leq J} n_j)^{-1/2} (\max_{\ell,k} |{\cal B}_{k,\ell}|)^{-1} a_0 (n\lambda)^\eta$ for some $\eta\in(0,1)$, we have
\[
\frac{1}{a_0 n\lambda}C_1 t_n \max_{j\in{\cal B}_{k,\ell}} n_j^{1/2}\rightarrow 0,\ \ {\rm as}\ \ n\rightarrow\infty,
\]
which indicates that  
\[
\ell_F(\vec{\gamma}, \vec{\sigma})-\ell_F(\vec{\gamma}^*, \vec{\sigma}) >0.
\]
This completes the proof. 
\end{proof}

\section*{SUPPLEMENTARY}

The Supplementary file provides additional simulation study. The preliminary and supplementary analyses of the rainfall application are related to Section 6 to support our main findings.

\subsection*{ACKNOWLEDGEMENTS} 
T. Y. was supported by Institute of Mathematics for Industry, Joint Usage/Research Center in Kyushu University (FY2024, Reference No. 2024a004). 
T. Y. was also financially supported by JSPS KAKENHI (Grant Nos. JP22K11935 and JP23K28043). 
S. K. was financially supported by JSPS KAKENHI (Grant Nos. JP23K11008, JP23H03352, JP23H00809, and JP25H01107).

\def\bibname{References}

\newpage

\setcounter{section}{0}
\renewcommand{\thesection}{S\arabic{section}}

\begin{center}
    {\LARGE \textbf{Supplementary file for \\ ^^ ^^ Structural grouping of extreme value models \\ via graph fused lasso''}}
    
    \vspace{1.5em}
    
    {\large Takuma Yoshida$^{1,3}$, Koki Momoki$^{2}$, and Shuichi Kawano$^{4}$}
    
    \vspace{1em}
    
    {\small
    $^{1,2}$\textit{Graduate School of Science and Engineering, Kagoshima University}\\
    $^{3}$\textit{Data Science and AI Innovation Research Promotion Center, Shiga University}\\
    $^{4}$\textit{Faculty of Mathematics, Kyushu University}
    }
\end{center}

\vspace{2em}

\section{Additional Simulation Study}

In this section, we designate GFL as the proposed grouping estimator constructed by graph fused lasso.

\subsection{Large sample size}

We consider a similar model to that presented in the main article. 
However, the sample size is $n=600$. 
Roughly speaking, in this setting, the cluster-wise estimator already has good performance. 
Figure \ref{FigSupSection1_1_1} shows the overall behavior and MSE ratio of the GFL and cluster-wise estimator. 
From the left panel, it is readily apparent that the difference between two estimators was quite small. 
The biases of the estimators were also small. 
In addition, the ratio of MSE appeared in right panel was close to one for all clusters. 
Consequently, when the sample size is large, the GFL has similar performance as cluster-wise estimator, which implies that the proposed method is useful as summarization method of clusters.

\begin{figure}[h]
\begin{center}
\includegraphics[width=150mm,height=60mm]{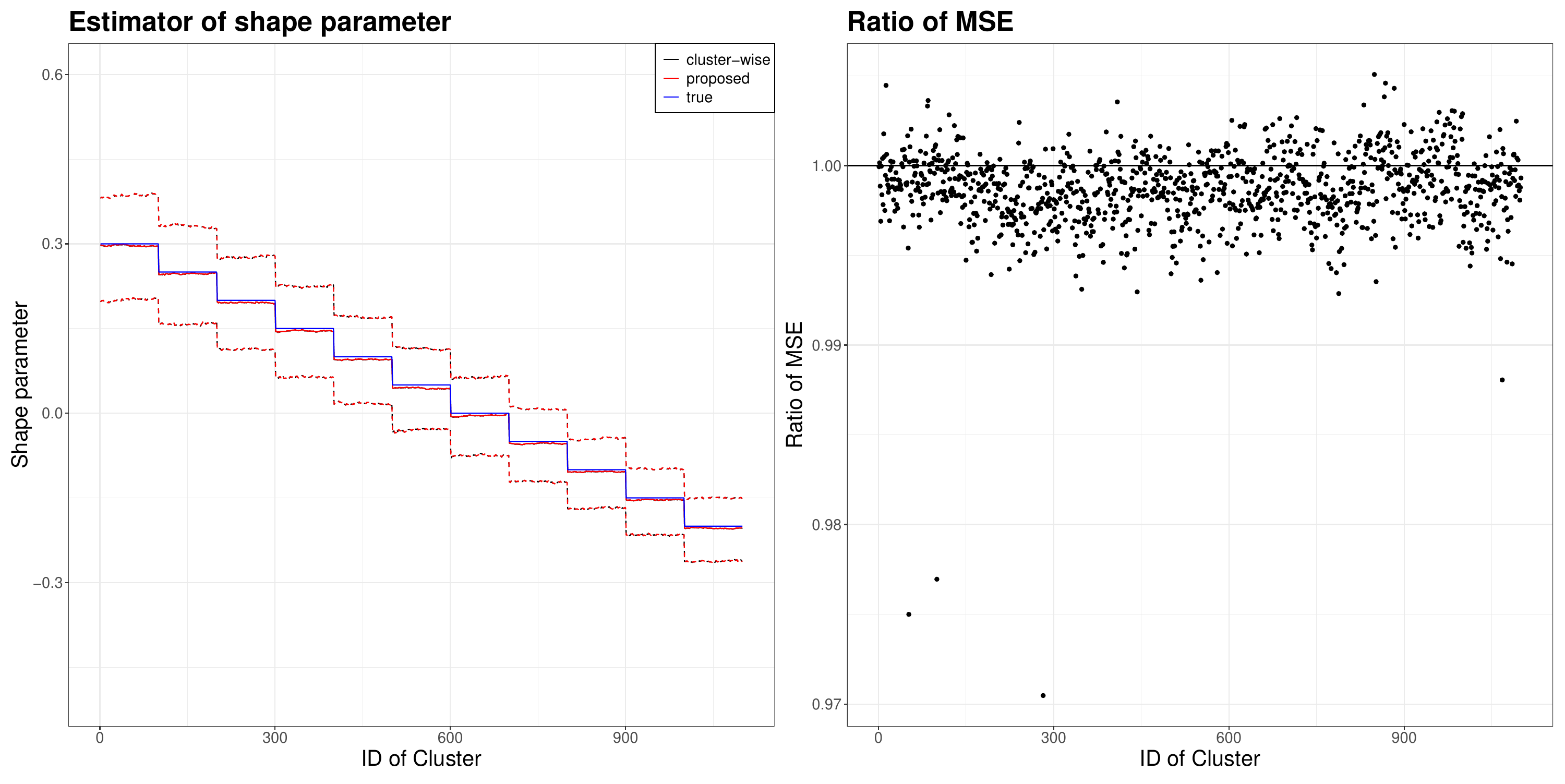}
\end{center}
\caption{Results of proposed estimator for each cluster $n=600$. 
Left: The cluster-wise estimator (black), the proposed estimator (red), and the true shape parameter (blue) for all clusters. 
The solid lines are medians. Dashed lines are lower/upper 5\% quantile. 
Right: ratio of MSE for all clusters $j=1,\ldots,J$. 
 \label{FigSupSection1_1_1}}
\end{figure}

In Figure \ref{FigSupSection1_1_2}, we presented the coverage probabilities and ratio of length of 95\% confidence interval (CI) of return level with period $\tau=1/(2n)$. 
From the left panel, the coverage probabilities constructed by GFL were lower than those by cluster-wise estimator for $j\leq 700$. 
These clusters have positive or zero shape parameters. 
On the other hand, for $j\geq 701$, where the shape parameter is negative, the GFL was not inferior to cluster-wise estimator 
In particular, for $j\geq 901$ ($\gamma_j \leq -0.1$), the GFL outperformed cluster-wise estimator on coverage probability. 
We can also confirm from right panel that the length of CI constructed rom GFL is short. 
Consequently, for large $n$, which is the case in which the estimation bias is small, the efficiency of the GFL was fully observed for clusters having negative shape parameters. 
We lastly note about the coverage probabilities for $601\leq j\leq 700$, which are small compared with other clusters. 
In these clusters, the true shape parameters are zero. 
However, since the estimates do not obtain exactly to zero, they tend to oscillate between positive and negative values. 
This oscillation is likely a primary factor in the reduction of coverage probabilities.

\begin{figure}
\begin{center}
\includegraphics[width=150mm,height=60mm]{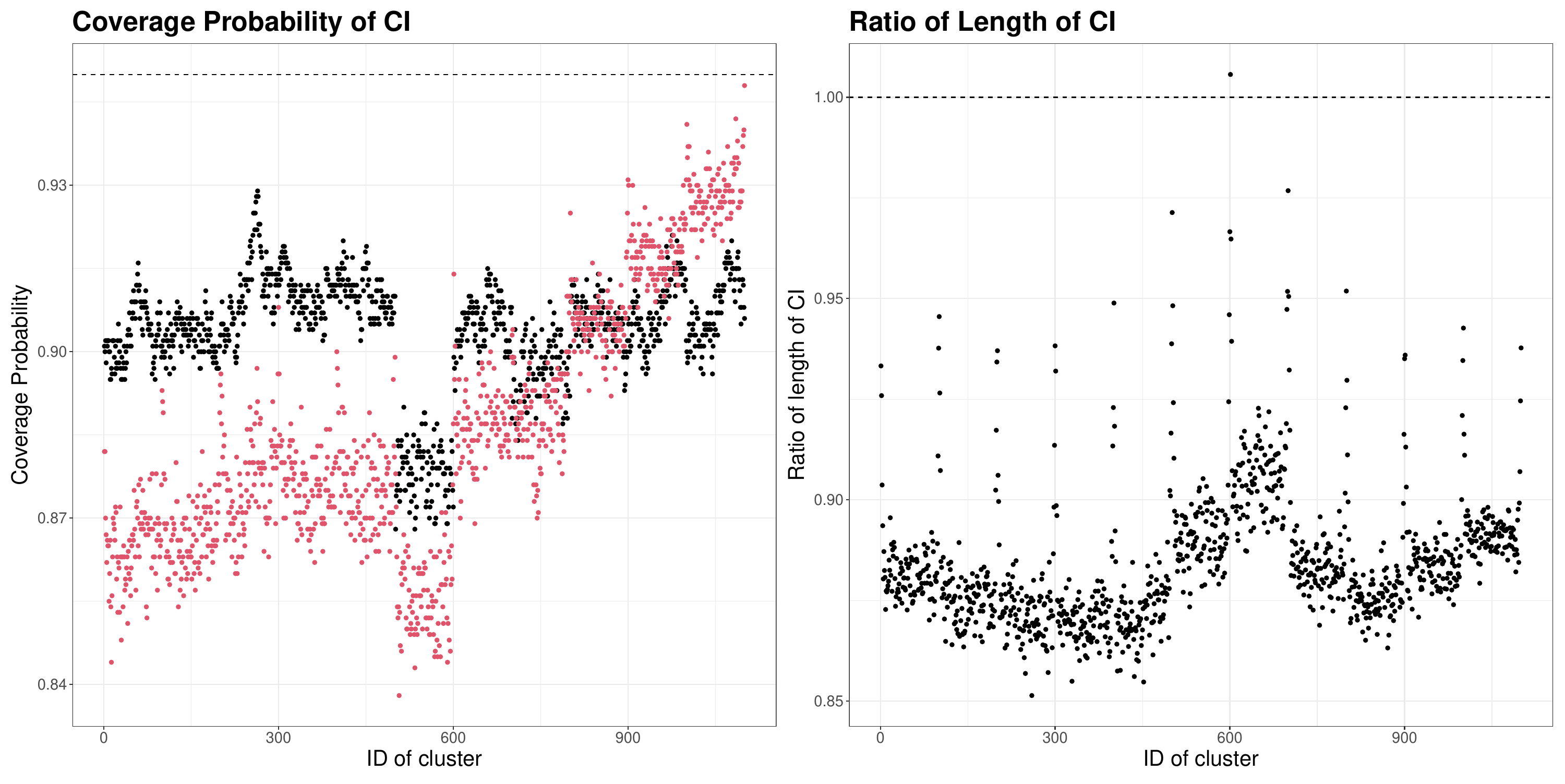}
\end{center}
\caption{ Results for 95\% confidence interval (CI) of return level for $j=1,\ldots,J$ for $n=500$. 
Left: Coverage probabilities from GFL (red) and cluster-wise method (black). 
Right: Average of ratio of length of CI of GFL over cluster-wise estimator.
\label{FigSupSection1_1_2}}
\end{figure}

\subsection{Homogeneity graph}

We considered a similar model to that given in the main article. 
The sample size was $n=120$. 
However, the graph structure of GFL was changed. 
In the main article, we used the graph under the situation that prior knowledge is useful. 
However, if we have no prior information between clusters, then the graph structure of GFL can only be constructed using the distance between cluster-wise estimators for each pair of clusters.
We defined the edge of the graph as 
\[
c_{j,k} = I(|\tilde{\gamma}_j-\tilde{\gamma}_k|<\delta)
\]
for some threshold $\delta>0$. 
For discussion in this section, $\delta$ is defined as the maximum value satisfying
\[
\sum_{j<k} I(|\tilde{\gamma}_j-\tilde{\gamma}_k|<\delta) <3000.
\]
Other settings are similar to those presented in Section 5 of the main article. 

The left panel of Figure \ref{FigSupSection1_2_1} shows the median, lower/upper 5\% quantile of the estimators, and true shape parameter for each cluster.
We see that both estimators have similar behavior. 
The ratios of MSE are shown in the right panel. 
For many clusters, MSE of GFL outperformed that of cluster-wise estimator. 
In particular, for $j\geq 601$ corresponding to zero and negative shape parameters, this tendency was remarkable. 
The result of MSE was drastically different from the results found for the estimator with a correctly specified graph presented in Figure 1 of main article.

\begin{figure}
\begin{center}
\includegraphics[width=150mm,height=60mm]{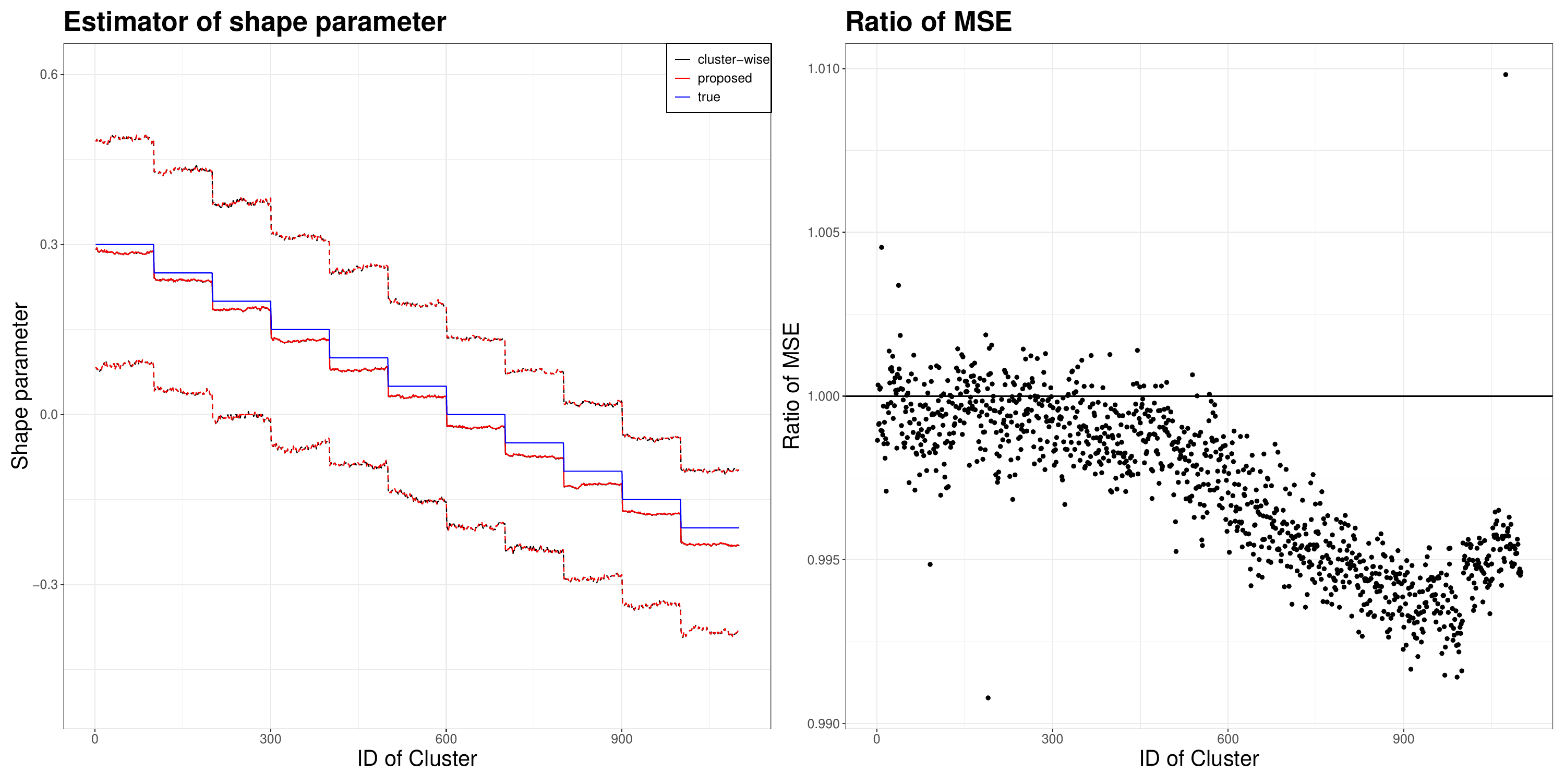}
\end{center}
\caption{Results of the proposed estimator with homogeneity graph for $j=1,\ldots,J$. 
Left: The cluster-wise estimator (black), the proposed estimator (red), and the true shape parameter (blue) for all clusters. 
The solid lines are medians. Dashed lines are lower/upper 5\% quantiles. 
Right: ratio of MSE for all clusters $j=1,\ldots,J$. 
 \label{FigSupSection1_2_1}}
\end{figure}

The results of the 95\% confidence intervals (CIs) for the return levels are presented in Figure \ref{FigSupSection1_2_2}. Their descriptions are similar to those of Figure \ref{FigSupSection1_1_2}.
This behavior is explainable by the length of the CI shown in the right panel.
Unfortunately, the coverage probability constructed by GFL was inferior to that by cluster-wise estimator for all clusters. 
This results indiate that the length of CI  from GFL was too short. 
Actually, we see from right panel that this consideration is true  for $j\leq 900$. 
These results suggest that the number of edges in the homogeneity graph should be controlled carefully. 
An overly dense graph might engender over-grouping and can degrade inferential performance.
Consequently, the findings of this example imply that we need to decide the graph structure carefully because it strongly affects the behavior of the estimator.

\begin{figure}
\begin{center}
\includegraphics[width=150mm,height=60mm]{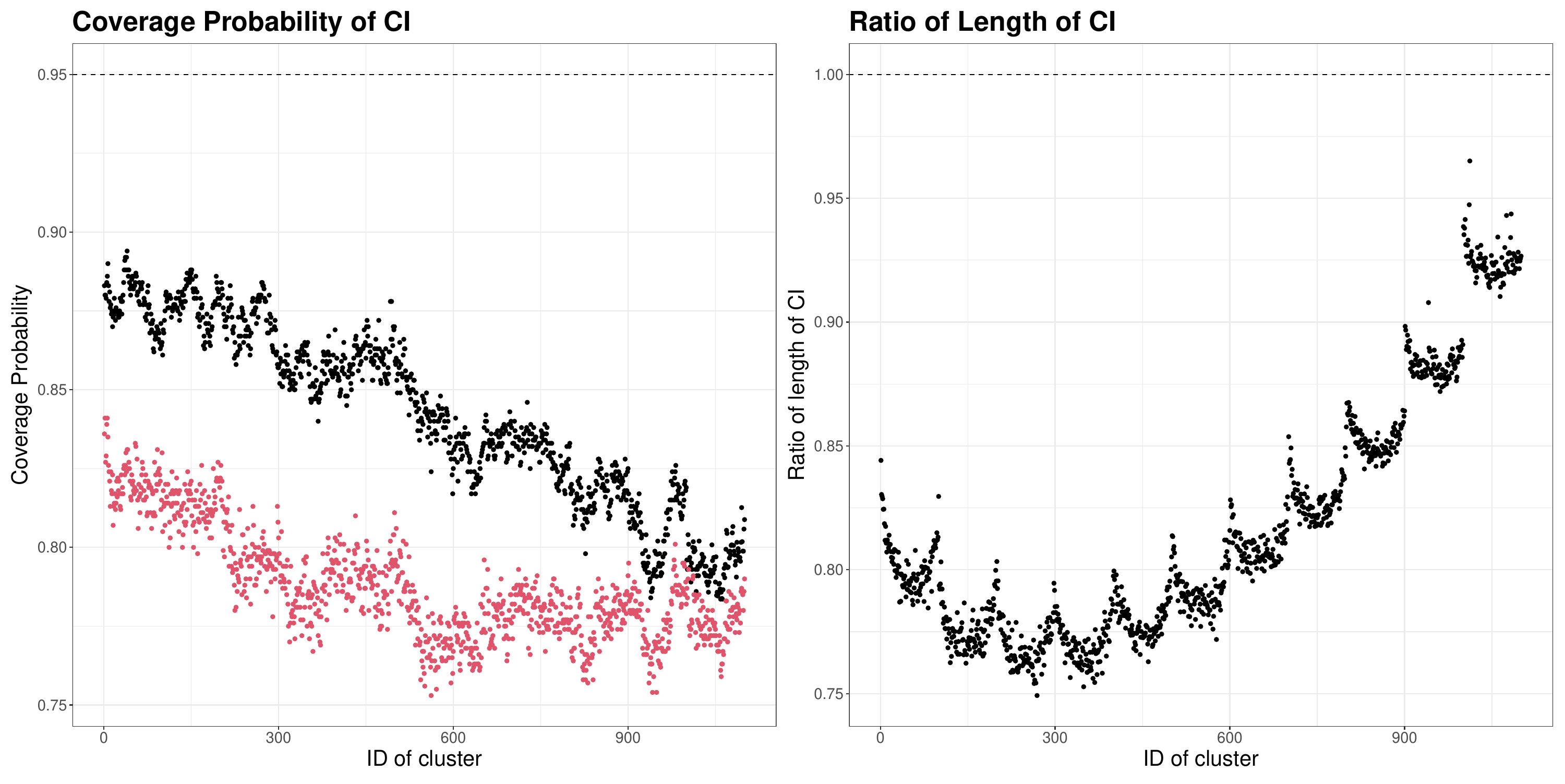}
\end{center}
\caption{ Results for 95\% confidence interval (CI) of return level for $j=1,\ldots,J$. 
Left: Coverage probabilities from proposed method with homogeneity graph (red) and cluster-wise method (black). 
Right: Average of ratio of length of CI of proposed estimator with homogeneity graph over the cluster-wise estimator. \label{FigSupSection1_2_2}}
\end{figure}

\subsection{Additional model with peak over threshold}

Because the true model used in main article is the fully GPD, the threshold selection was not needed. 
As described in this section, we consider a more complex model to elucidate the sensitivity of the threshold. 
For $i=1,\ldots,n$, the uniform random vector $(U_{i,1},\ldots,U_{i,J})$ was generated from a similar method to that of Section 5 of the main article. 
Let $w_0=\Phi^{-1}(0.95)\approx 1.64$, where $\Phi$ is the distribution function of standard normal. 
We then generated $X_{ij}$ from the distribution as
\begin{eqnarray*}
X_{ij}
=
\left\{
\begin{array}{cc}
\Phi^{-1}(U_{ij}) & U_{ij}<0.95,\\
w_0 + F_j^{-1}\left(\frac{U_{ij}-0.95}{0.05}\right) & U_{ij}\geq 0.95,
\end{array}
\right.
\end{eqnarray*}
where $F_j$ is the generalized Pareto distribution (GPD) given in Section 5 of the main article. 
The distribution function presented above is the type of mixed distribution. 
In this setting, $w_0$ is apparent as the true threshold. 
For this example, we need to use the peak over threshold method to fit GPD to investigate the tail behavior of the distribution.
We now set the threshold as $w_j = w_0$ for all $j=1,\ldots,J$. 
In this simulation, the sample size is set as $n=2400$.
Then, the number of effective sample $n_j$ is random and its rate distributes to the binomial with sample size $n=2400$ and occurrence rate $p=0.95$. 
Therefore, we have $E[n_j]=120$, which is similar to the effective sample size in the setting of main article, and $\sqrt{V[n_j]}\approx 10.67$. 
Consequently, for each $j$, we obtain $n_j\in[100, 140]$ with probability about 0.95. 
We applied the GFL to the extreme value data observed from the setting above.
In our method, the graph structure of fused lasso was set as similar to that given in Section 5 of the main article. 

Figure \ref{FigSupSection1_3_1} presents the overall behavior of the GFL and the cluster-wise estimator (left panel), as well as the ratio of the MSE of the two estimators (right panel).
From the left panel, the GFL exhibited behavior very similar to that of the cluster-wise estimator across all clusters, indicating that the proposed method can serve as an effective summarizing approach even for the peak over-threshold setting.
From right panel, for many clusters, the ratio was less than one, demonstrating an overall improvement in estimation accuracy.

\begin{figure}
\begin{center}
\includegraphics[width=150mm,height=60mm]{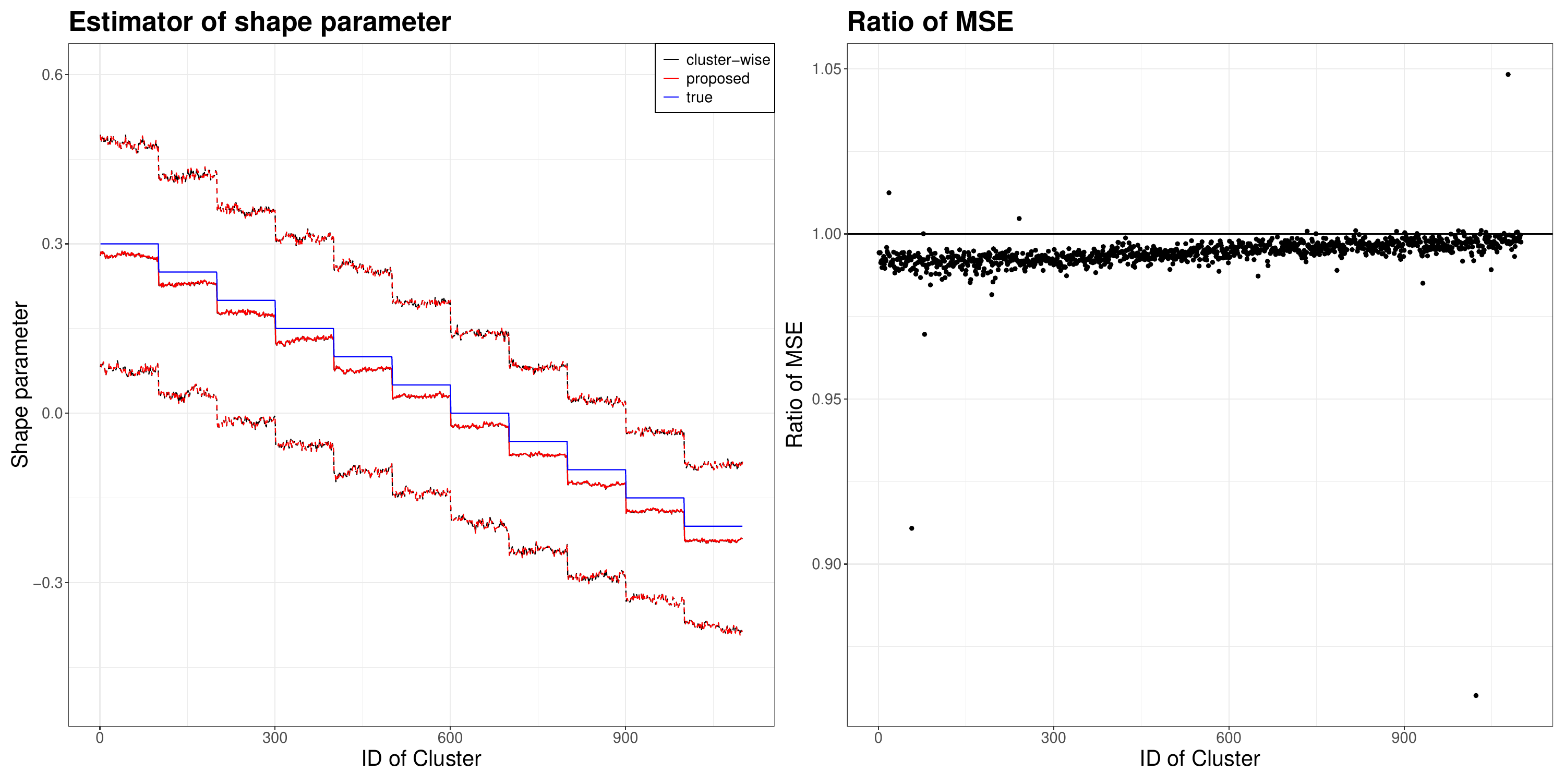}
\end{center}
\caption{Results of proposed estimator for $j=1,\ldots,J$ for mixed model. 
The description is similar to that for Figure \ref{FigSupSection1_2_1}. 
 \label{FigSupSection1_3_1}}
\end{figure}

Figure \ref{FigSupSection1_3_2} describes results of 95\% CI of return level for mixed distribution. 
The descriptions are similar to Figure \ref{FigSupSection1_1_2}.  
Consequently, for several clusters, the coverage probability of CI constructed by the GFL was improved. 
This result clearly demonstrated the efficiency of the GFL. 
Nevertheless, we see from right panel the length of CI by GFL was smaller than that by cluster-wise estimator for many clusters. 
This result and right panel of Figure \ref{FigSupSection1_3_1} indicates the GFL improves not only variance but also bias. 
Consequently, we were able to confirm that the proposed method is more useful when the true model is not full GPD.

\begin{figure}
\begin{center}
\includegraphics[width=150mm,height=60mm]{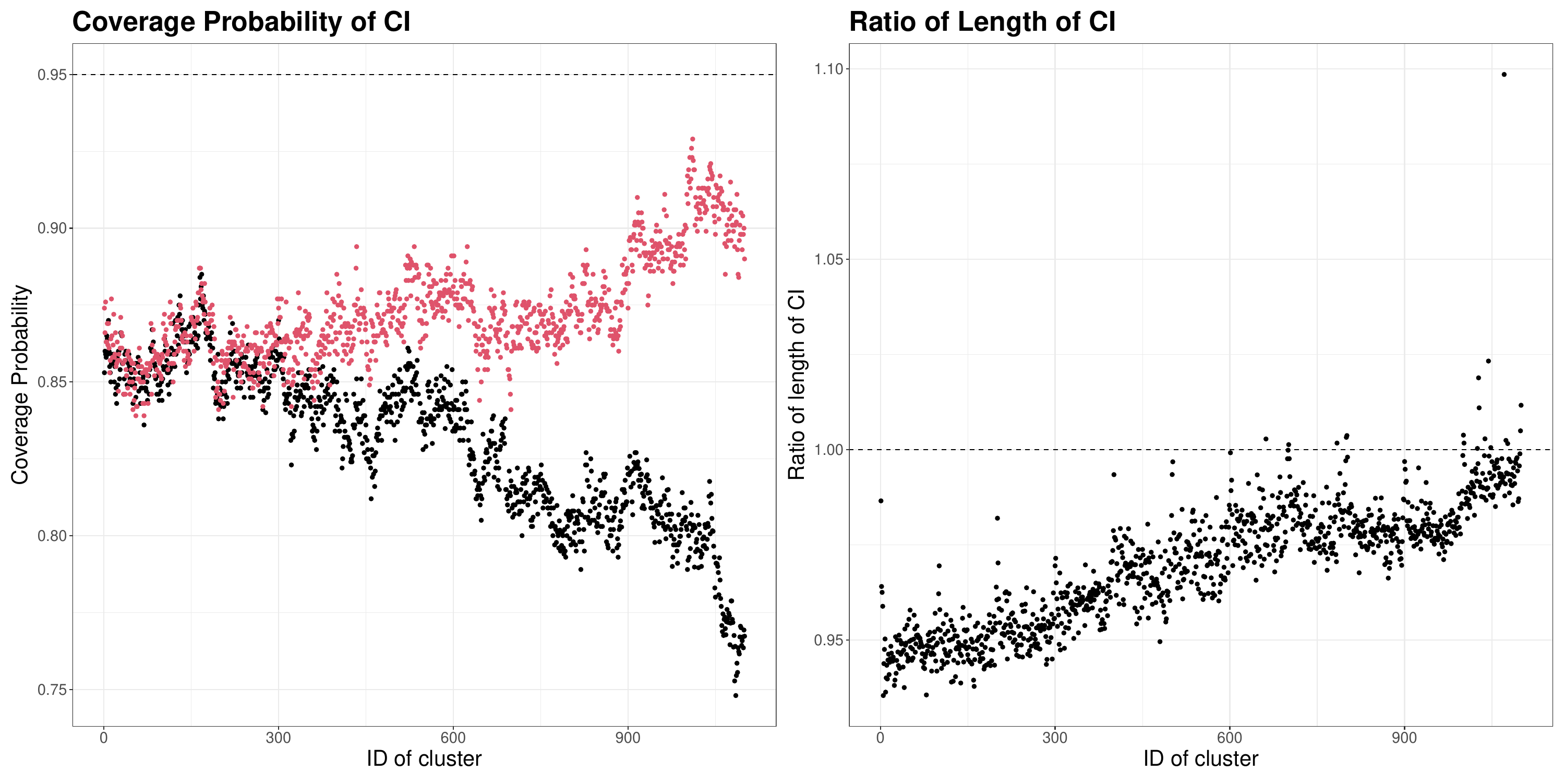}
\end{center}
\caption{ Results for 95\% confidence interval (CI) of return level for $j=1,\ldots,J$ for mixed model. 
The description is similar to that given for Figure \ref{FigSupSection1_1_2}.  \label{FigSupSection1_3_2}}
\end{figure}

\section{Additional analysis in application}

In Section 6 of the main article, we presented the results of applying our proposed method to rainfall data observed at multiple sites in Japan. This section reports preliminary and supplementary analyses that support these applications.

\subsection{Threshold selection}

When we apply the GPD to the data, the threshold selection is important. 
Because the climate condition varies across sites, the threshold must be selected appropriately for each site.
Meanwhile, because our goal is to group the tail behavior of rainfall for sites, we aim to balance the conditions across clusters to the greatest extent possible. 
To this ends, we chose the thresholds for clusters so that the effective sample sizes are equal for all clusters.
We describe how to choose the effective sample size (and threshold). 
Let $X_{(1),j}\geq\cdots\geq X_{(n),j}$ be order statistics from $X_{1,j},\ldots,X_{n,j}$ for $j=1,\ldots,J$. 
Then, for the effective sample size $k$, we define the risk as 
\[
R(k) = \frac{1}{J}\sum_{j=1}^J \frac{1}{k}\sum_{i=1}^k \{X_{(i),j}-Q_j(p_i)\}^2,
\]
where $p_i =i/(k+0.5)$ and 
\[
Q(p) =\frac{\tilde{\sigma}_j}{\tilde{\gamma}_j(\tilde{\gamma}_j+1)}\{p^{-\tilde{\gamma}_j}-1\}.
\]
This risk is evident in the QQ-plot errors of the GPD. 
We define the optimal number of exceedances, $k=k_{opt}$, as the stable point along the path of $R(k)$. 
Then, for each cluster, the threshold $w_j$ is determined so that $n_{0j}=\sum_{i=1}^n I(X_{ij}>w_j) = k_{opt} (j=1,\ldots,J)$. 

\begin{figure}
\begin{center}
\includegraphics[width=150mm,height=80mm]{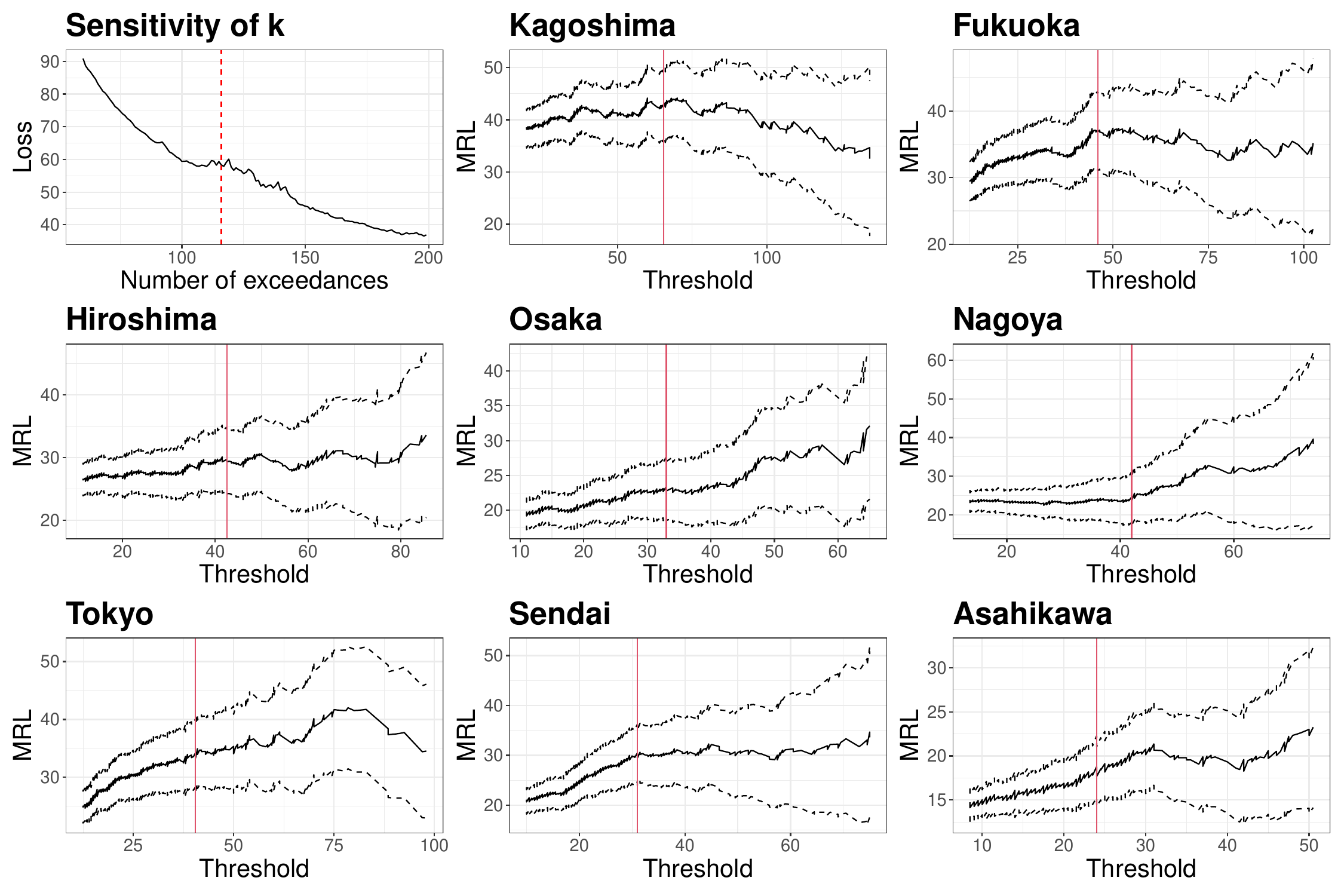}
\end{center}
\caption{Top left: Sensitivity of $R(k)$ over $k$. Remaining panels: Mean residual life plot of selected clusters (sites) with selected threshold (red). \label{FigSupSection2_1}}
\end{figure}

The top left panel of Figure \ref{FigSupSection2_1} illustrates the path of $R(k)$. 
The red vertical line shows stable point, which was $k_{opt}=116$. 
This corresponds to about the 0.967th quantile of data for each cluster.
We next check the validity of the threshold using a mean residual life plot (MRL), as described in Chapter 4 of Coles (2001).
We now specifically examine the eight sites presented in Figure 1 of the main text: Kagoshima, Fukuoka, Hiroshima, Osaka, Nagoya, Tokyo, Sendai, and Asahikawa. 
Panels other than that at the top left in Figure \ref{FigSupSection2_1} present the MRL plot and the 95\% confidence interval for these eight sites. 
For all eight sites, the MRL remains stable around the threshold shown by the red vertical line.

Figure \ref{FigSupSection2_2} presents the spatial distribution of threshold values across all clusters, revealing that these values are highly sensitive to geometric characteristics.

\begin{figure}
 \centering
\includegraphics[width=0.45\linewidth]{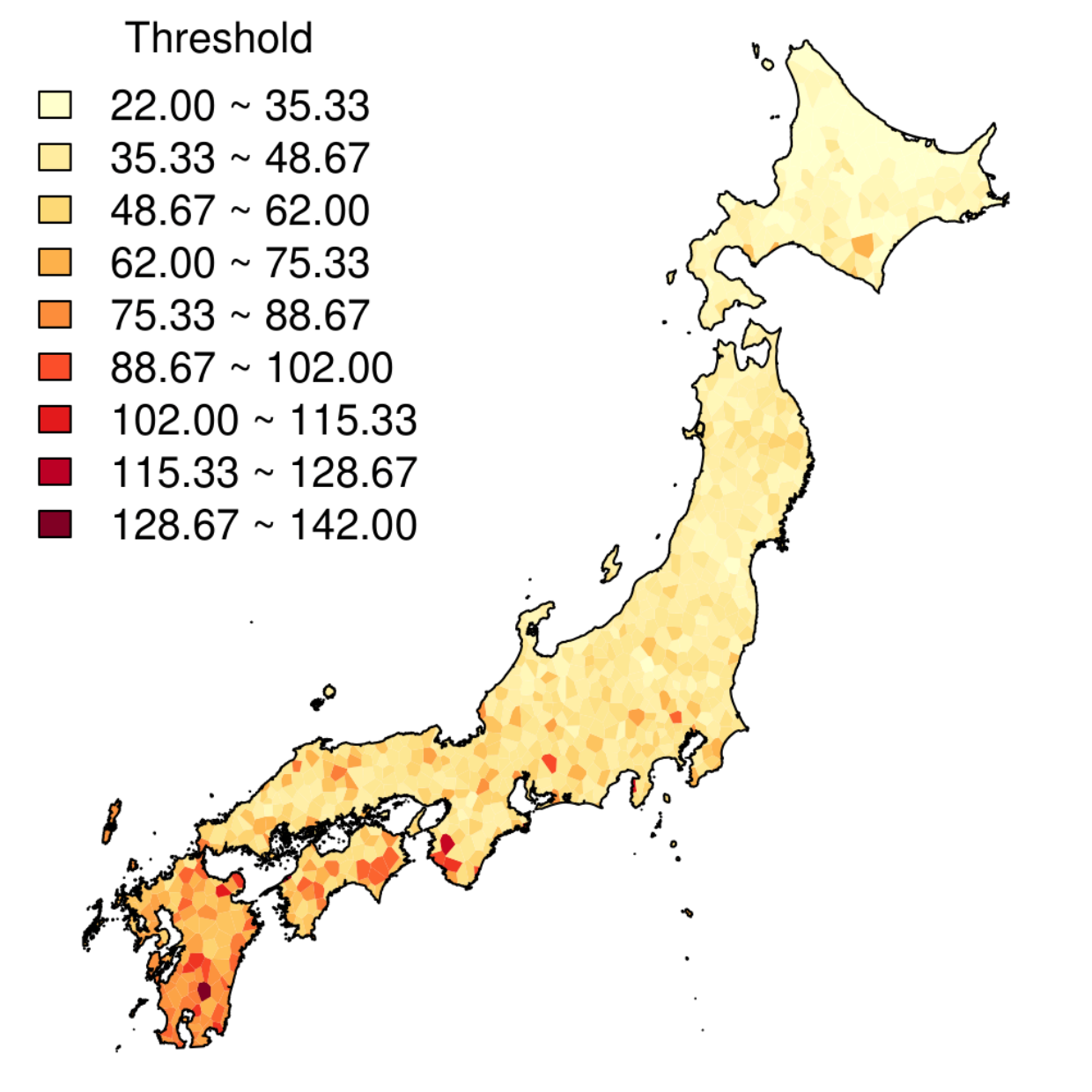}
\caption{Threshold value for all clusters. \label{FigSupSection2_2}}
\end{figure}

\subsection{Information of graph}

In our method, making the graph of the fused lasso is important.
In Section 6 of the main article, the graph is constructed by the asymptotic dependence measure of each pair of the clusters. 
We now confirm which pairs of clusters are connected by graph fused lasso. 
The left panel of Figure \ref{FigSupSection2_34} shows the asymptotic dependence between the precipitation of Kagoshima and other sites (left) and Sendai and other sites (right). 
It is readily apparent that the strength of the asymptotic dependence captures the distance between clusters. 
Although it was not shown here, for any site, the magnitude of the asymptotic dependence is strongly related to the distance between sites, with closer sites showing greater asymptotic dependence. 

\begin{figure}[h]
\centering
  \begin{minipage}{0.45\linewidth}
    \centering
    \includegraphics[width=\linewidth]{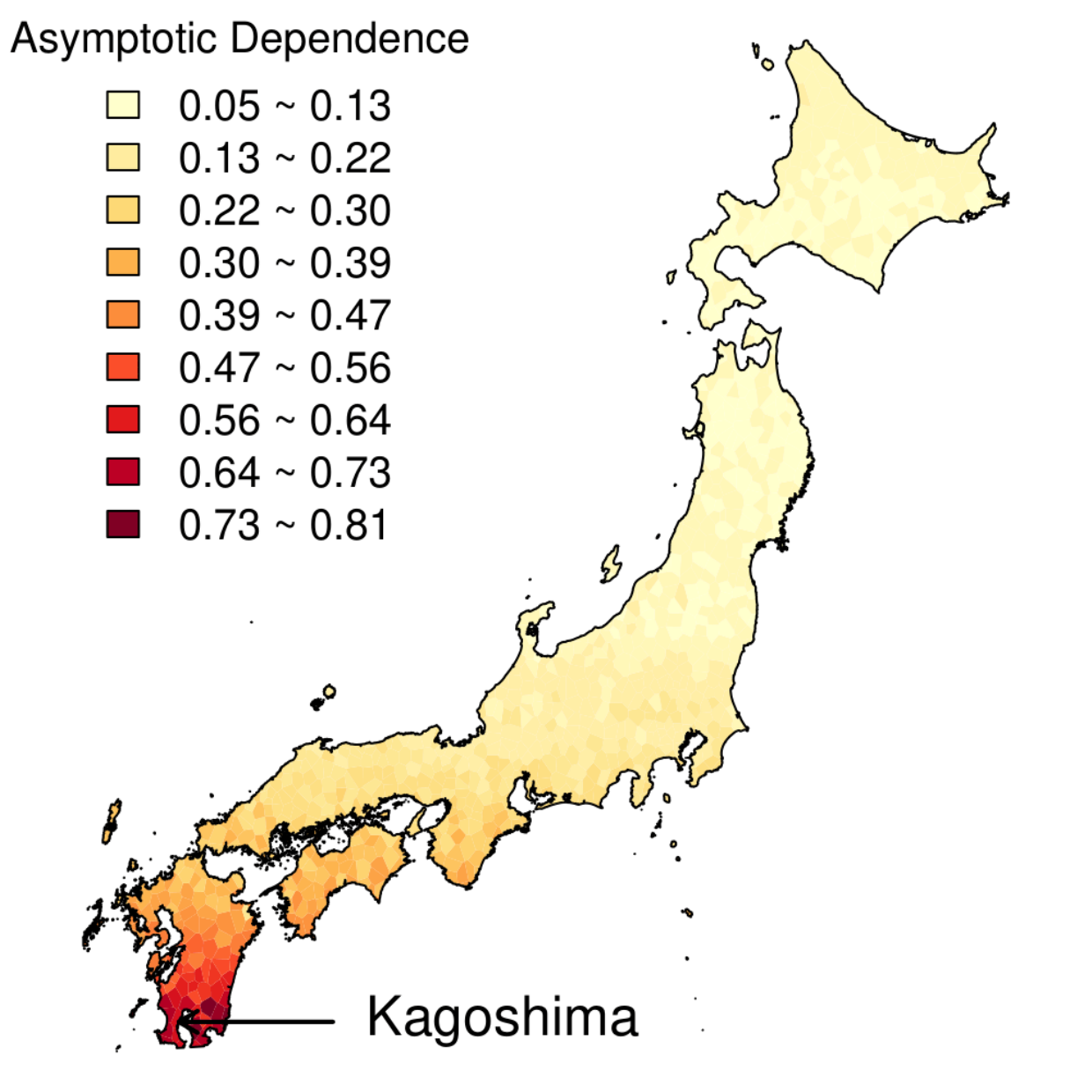}
    \vspace{0.5\baselineskip}
    \textbf{Kagoshima and other sites}
  \end{minipage}
  \hfill
  \begin{minipage}{0.45\linewidth}
    \centering
    \includegraphics[width=\linewidth]{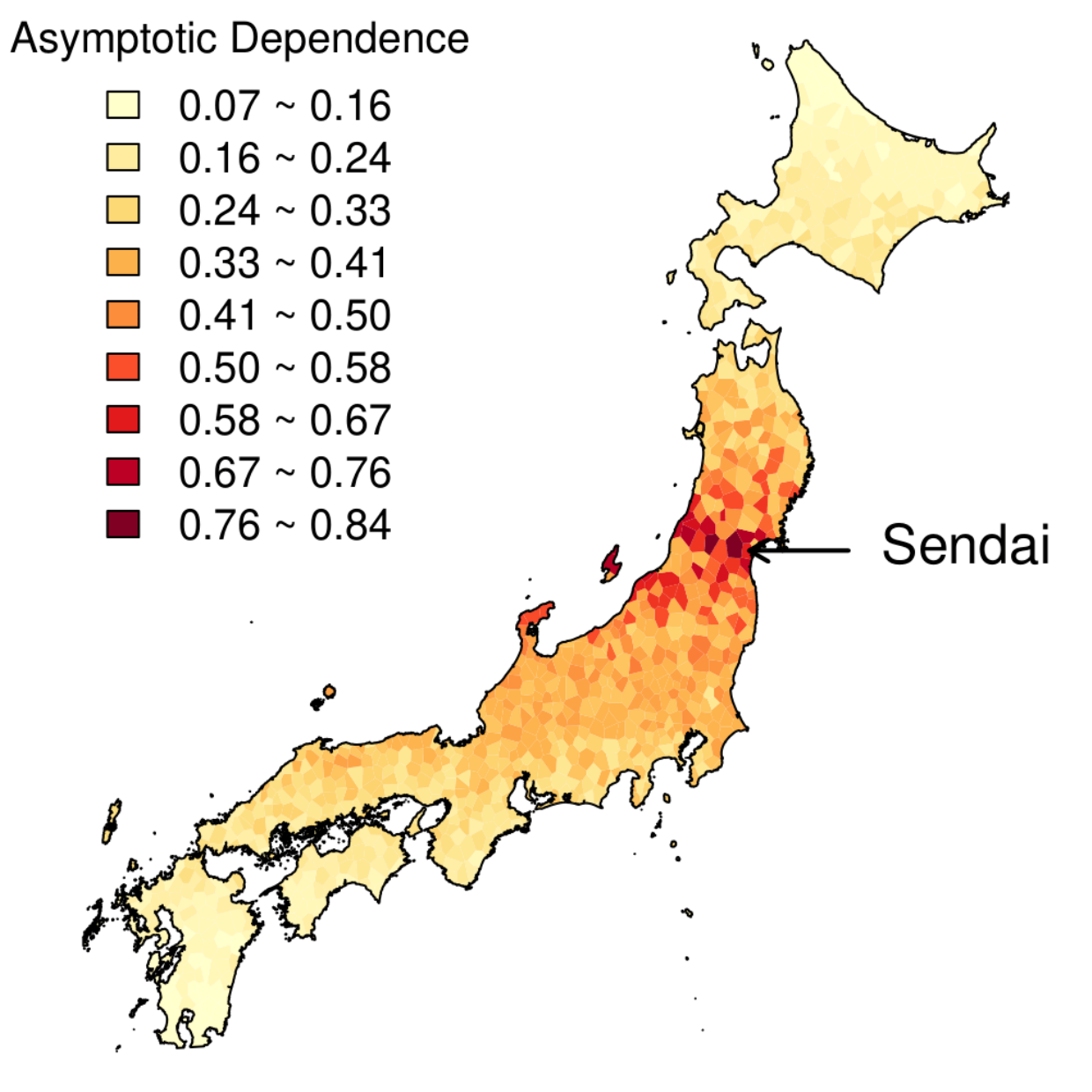}
    \vspace{0.5\baselineskip}
    \textbf{Sendai and other sites}
  \end{minipage}
\caption{Asymptotic dependence for each pair of Kagoshima and other sites (left) and Sendai and other sites (right).}
\label{FigSupSection2_34}
\end{figure}

In the main article, the edge is added to the pairs that the asymptotic dependence over cut-off $c^*=0.76$. 
Then, the total number of edges was 849. 
We next confirm the sensitivity of the cut-off.
The left panel of Figure \ref{FigSupSection2_5} shows the value of $c^*$ and the number of edges. 
For example, when the cut-off value is $0.65$, the number of edges is greater than 6000, whereas the number of edges is smaller than 5 when the cut-off is larger than 0.85. 
Consequently, $c^*$ is crucial to the graph size for grouping clusters. 
In the right panel, we present the BIC value versus cut-off of the asymptotic dependence using the graph of the fused lasso. 
The value $c^*=0.76$ attained the minimum value of BIC in this example.

 \begin{figure}
\begin{center}
\includegraphics[width=150mm,height=60mm]{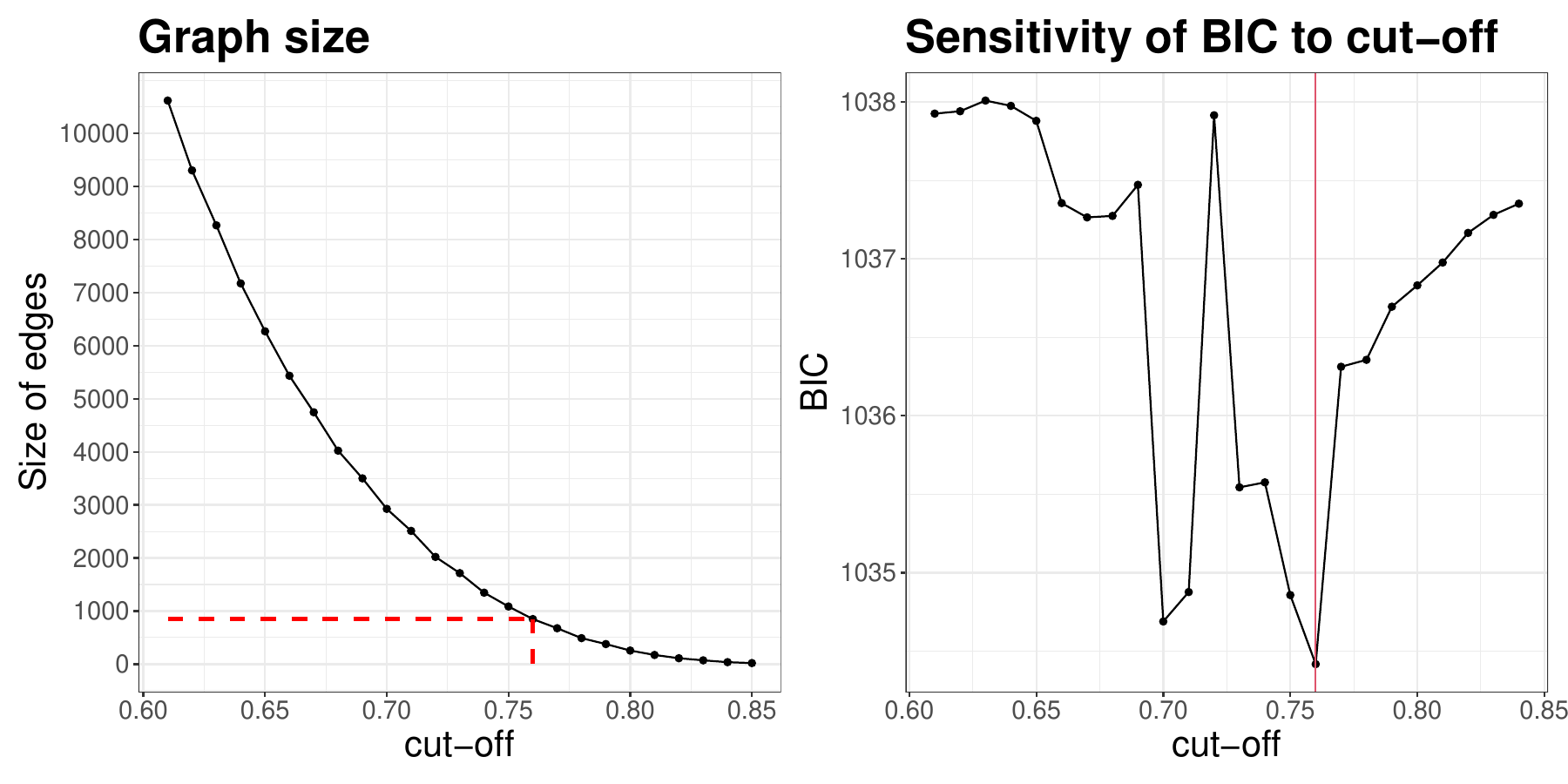}
\end{center}
\caption{Left: Number of edges for each cut-off value. Right: BIC versus cut-off. \label{FigSupSection2_5}}
\end{figure}

The left panel of Figure \ref{FigSupSection2_6} shows the BIC value against with the number of new group $K=K(\lambda)$ with $c^*=0.76$. 
The minimum BIC is attained when $\lambda=0.368$ and $K=K(\lambda) = 293$. 
We can see from the right panel of Figure \ref{FigSupSection2_6} how many sites were accumulated within each new group. 
The largest group included 68 sites, and this group contains Hiroshima. 
The second, third and fourth largest included 36, 22 and 20. 
Fukuoka belongs to the fourth largest group.
Among the resulting groups, the case in which two sites were merged was the most frequent; 54 new groups were formed in this manner.

\begin{figure}
 \centering
\includegraphics[width=150mm,height=60mm]{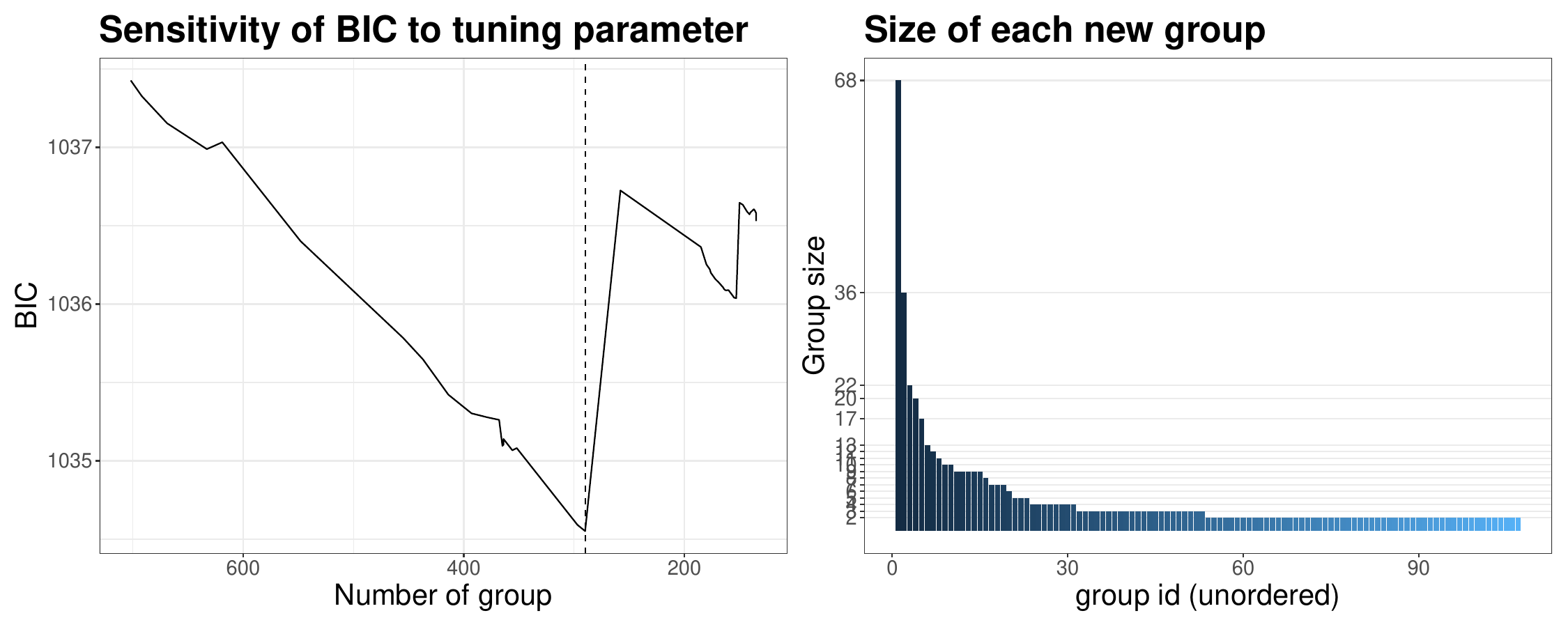}
\caption{Result of proposed grouping method. 
Left: The BIC value against with number of new group. The vertical line shows estimated number of new group.
Right: The number of clusters within each new group. \label{FigSupSection2_6}
}
\end{figure}

\subsection{Additional results of return level}

\begin{figure}[h]
\centering
  \begin{minipage}{0.45\linewidth}
    \centering
    \includegraphics[width=\linewidth]{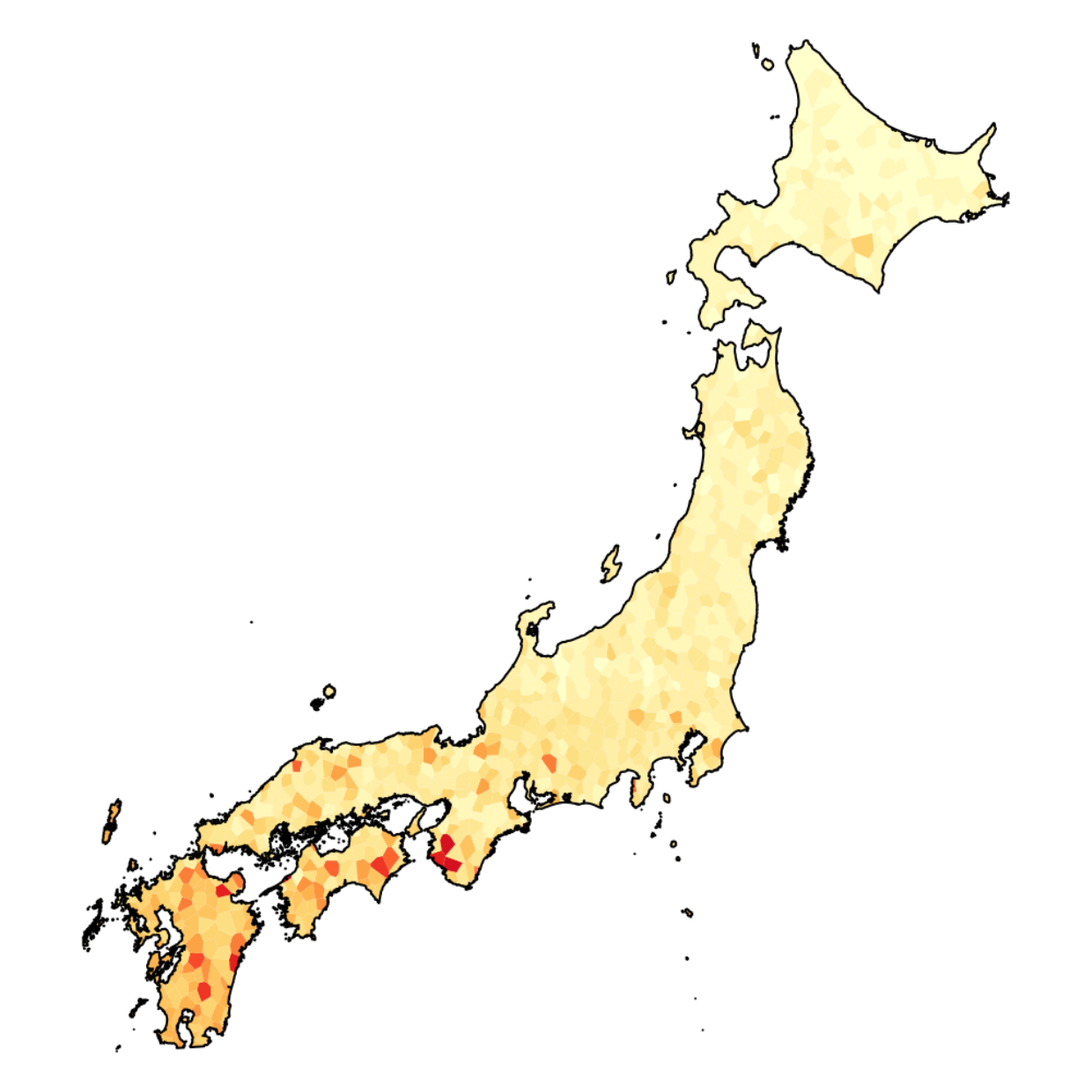}
    \vspace{0.5\baselineskip}
    \textbf{Cluster-wise}
  \end{minipage}
  \hfill
  \begin{minipage}{0.45\linewidth}
    \centering
    \includegraphics[width=\linewidth]{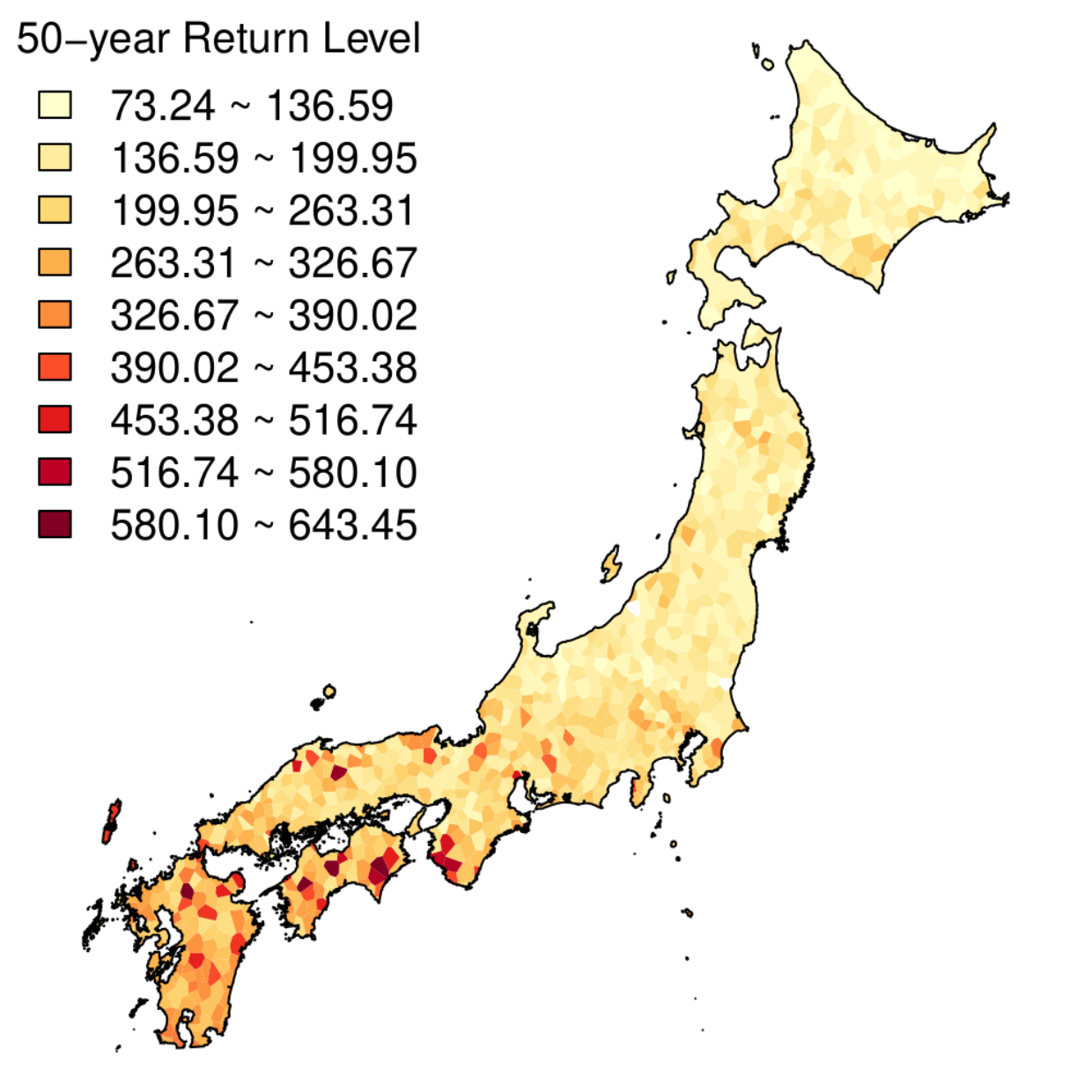}
    \vspace{0.5\baselineskip}
    \textbf{Graph fused lasso}
  \end{minipage}
\caption{50-year return level for all clusters constructed by cluster-wise estimator (left) and proposed estimator (right)}
\label{FigSupSection2_78}
\end{figure}

Figure \ref{FigSupSection2_78} shows a map of 50-year return levels for all clusters. 
Left and right panels respectively correspond to the cluster-wise and proposed method. 
Overall, the return levels increased across many sites. 
This can be interpreted as the estimates at many sites being pulled upward due to the influence of neighboring high-risk sites and incorpolates these features.

 \begin{figure}
\centering
\includegraphics[width=0.45\linewidth]{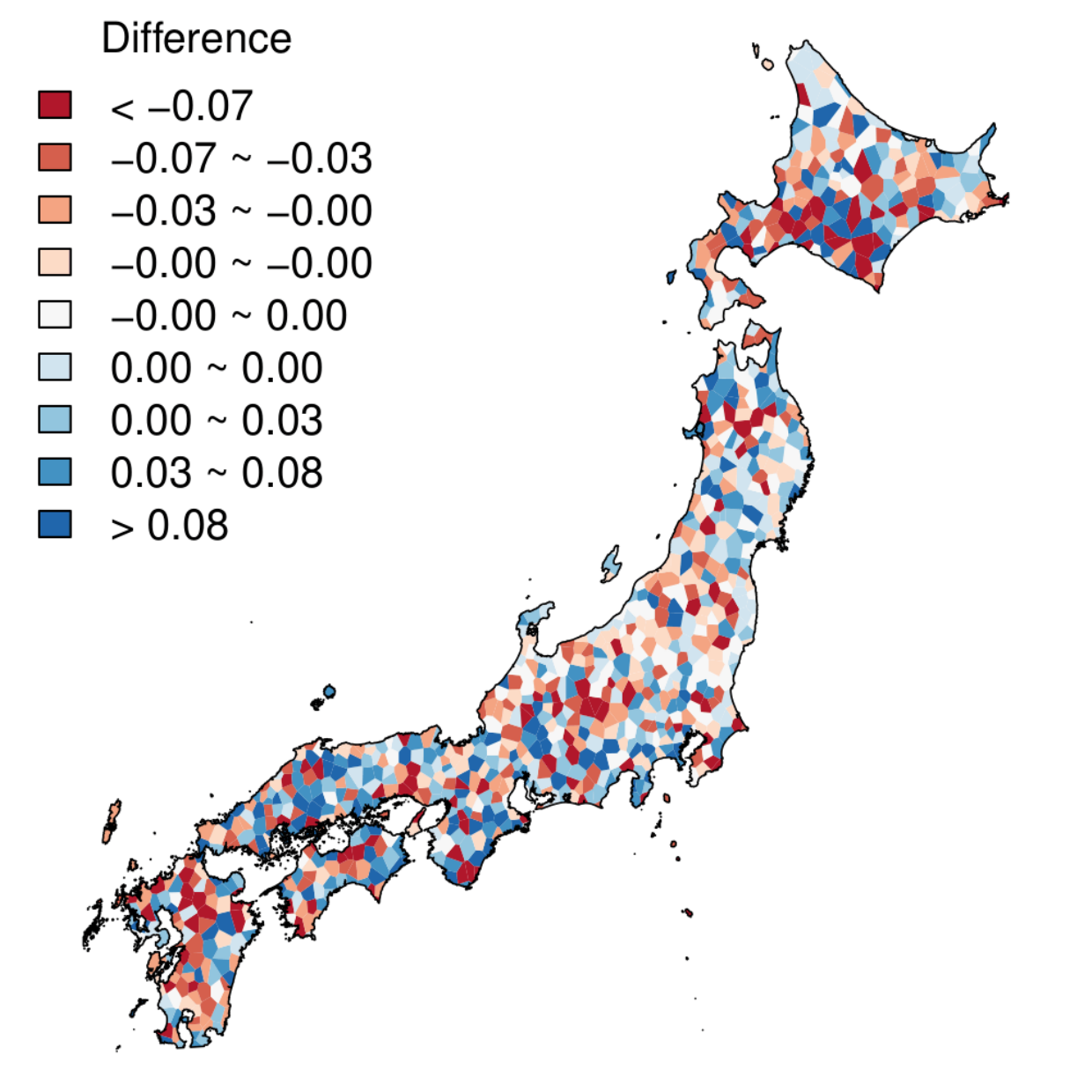}
\caption{Difference between shape parameter constructed using the cluster-wise estimator and the proposed estimator for all clusters. \label{FigSupSection2_9}}
\end{figure}

Figure \ref{FigSupSection2_9} depicts the image map of the difference between estimator of shape parameter constructed by cluster-wise method and proposed grouping method for all clusters. 
Clusters shown in red or blue correspond to locations at which grouping is applied, whereas white clusters denote no grouping.
The colors reflect the direction and magnitude of differences between estimators.
Overall, the results exhibit a sparse pattern, suggesting that the proposed method tends to induce grouping across sites, leading to increased homogeneity.

\subsection{Proposed method with homogeneity graph} 

Here, we consider the homogeneity based graph as the other type of graph (Ke et al. 2015). 
That is, the edge is added to $j$-th and $k$-th clusters if $|\tilde{\gamma}_j-\tilde{\gamma}_k|<\delta$ for some $\delta>0$. 
For analyses in this section, we set $\delta=0.005$. 
Then, the total number of edges is 5043. 

 \begin{figure}
\centering
\includegraphics[width=0.45\linewidth]{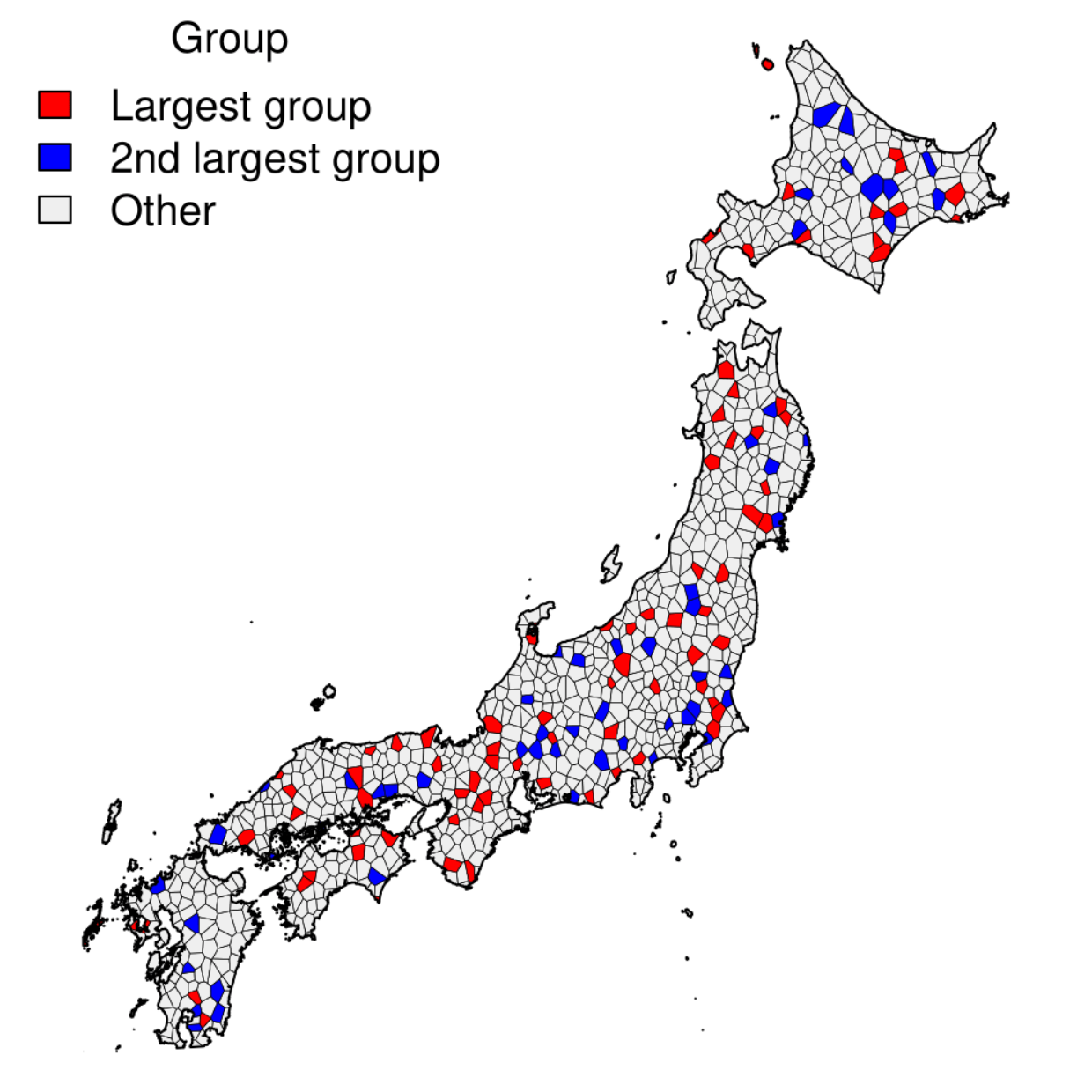}
\caption{Clusters grouped using the proposed method with the homogeneity graph. \label{FigSupSection2_10}}
\end{figure}

In this setting, we apply the proposed method to the same data as used in Section 6.
As a result, the respective sizes of the largest and the second largest groups are 60 and 43.
In Figure \ref{FigSupSection2_10}, clusters belonging to these two groups are indicated.
Clusters marked with the same color (red or blue) are grouped as having the same shape parameter.

The results showed that the obtained grouping did not reflect the spatial structure of sites.
Consequently, the tail behavior of each cluster becomes difficult to interpret.
Moreover, the resulting grouping appears to be influenced by estimation error.

These findings suggest that incorporating a user-specified graph based on prior information is important for obtaining a reasonable interpretation.
At least in cases where prior knowledge such as tail dependence or spatial distance between clusters is available, such information would be useful to construct the graph structure of the fused lasso.

\def\bibname{References}

\end{document}